\RequirePackage{fix-cm}

\documentclass[twocolumn]{svjour3}
\smartqed
\usepackage{graphicx}
\usepackage[hidelinks]{hyperref}
\usepackage{tabularx}
\usepackage{booktabs}
\usepackage{amsmath}
\usepackage{amsfonts}
\usepackage{algorithm}
\usepackage{algpseudocode}
\usepackage{pifont}
\usepackage{multirow}
\usepackage[table,xcdraw]{xcolor}
\usepackage{subcaption}
\usepackage{tablefootnote}
\usepackage[numbers,sort]{natbib}
\usepackage[inline]{enumitem}
\usepackage[normalem]{ulem}

\usepackage[final,todonotes={textsize=tiny}]{changes}
\setcommentmarkup{{\IfIsColored{\color{authorcolor}}{}}}
\sethighlightmarkup{{\IfIsColored{\color{authorcolor}}{}}}
\definechangesauthor[name={R1}, color=blue]{R1}
\definechangesauthor[name={R2}, color=blue]{R2}
\definechangesauthor[name={R3}, color=blue]{R3}

\newcommand{\quotes}[1]{`#1'}
\newcommand{\dquotes}[1]{``#1''}

\newcommand{\dbpedia}{DBPedia}
\newcommand{\stringdb}{STRING}

\newcommand{\avantgraph}{\texttt{AvantGraph}}

\newcommand{\avantgraphu}{\texttt{AvantGraph$_u$}}
\newcommand{\agu}{\texttt{AG$_u$}}
\newcommand{\avantgraphs}{\texttt{AvantGraph$_o$}}
\newcommand{\ags}{\texttt{AG$_o$}}
\newcommand{\avantgraphwg}{\texttt{AvantGraph$_s$}}
\newcommand{\agwg}{\texttt{AG$_s$}}

\newcommand{\millenniumdb}{\texttt{MillenniumDB}}
\newcommand{\mdb}{\texttt{MDB}}
\newcommand{\millenniumdbv}{\texttt{MillenniumDB 0.2}}

\newcommand{\duckdb}{\texttt{DuckDB}}
\newcommand{\ddb}{\texttt{DDB}}
\newcommand{\duckdbv}{\texttt{DuckDB 0.8.1}}

\newcommand{\postgres}{\texttt{PostgreSQL}}
\newcommand{\pg}{\texttt{PG}}
\newcommand{\postgresv}{\texttt{PostgreSQL 16.0}}

\newcommand{\virtuoso}{\texttt{Virtuoso}}
\newcommand{\vt}{\texttt{VT}}
\newcommand{\virtuosov}{\texttt{Virtuoso 7.2}}

\newcommand{\kuzu}{\texttt{K\`uzu}}
\newcommand{\umbra}{\texttt{Umbra}}
\newcommand{\duckpgq}{\texttt{DuckPGQ}}
\newcommand{\tigergraph}{\texttt{TigerGraph}}

\newcommand{\ebp}[1]{\widehat{p}_{#1}}
\newcommand{\abp}[1]{\overline{p}_{#1}}
\newcommand{\ebpu}{\ebp{u}}
\newcommand{\ebpo}{\ebp{o}}
\newcommand{\abpu}{\abp{u}}
\newcommand{\abpo}{\abp{o}}

\newcommand{\timeplan}[1]{t_{\texttt{opt}}(#1)}
\newcommand{\timeexec}[1]{t_{\texttt{exe}}(#1)}

\newcommand{\uq}{\mathbb{U}_Q}
\newcommand{\sq}{\mathbb{O}_Q}

\newcommand{\scra}[2]{\overset{\scriptscriptstyle \rightarrow}{#1}\vphantom{#1}^{\scriptscriptstyle #2}}
\newcommand{\scrp}[2]{\overset{\scriptscriptstyle \leftarrow}{#1}\vphantom{#1}^{\scriptscriptstyle #2}}

\newcommand{\potcard}{PC}
\newcommand{\pottime}{PT}
\newcommand{\actcard}{AC}
\newcommand{\acttime}{AT}

\algtext*{EndWhile}
\algtext*{EndIf}
\algtext*{EndFor}
\algrenewcommand\algorithmicindent{1.0em}

\newcounter{qeqc}
\newcounter{deqc}

\makeatletter
\newenvironment{qequation}
{\stepcounter{qeqc}
 \def\Hy@chapterstring{equation}
 \def\Hy@chapapp{qequation}
\equation}
{\endequation}
\makeatother

\makeatletter
\newenvironment{dequation}
{\stepcounter{deqc}
 \def\Hy@chapterstring{equation}
 \def\Hy@chapapp{dequation}
\equation}
{\endequation}
\makeatother

\begin{document}

\setlength{\abovedisplayskip}{3pt}
\setlength{\belowdisplayskip}{3pt}

\title{Optimizing Navigational Graph Queries}

\author{Thomas Mulder \and
        George Fletcher \and 
        Nikolay Yakovets
}

\institute{
    T. Mulder (0009-0002-5847-5941) \at
    Technische Universiteit Eindhoven, Eindhoven, Netherlands
    \email{t.mulder@tue.nl}
    \and
    G. Fletcher (0000-0003-2111-6769) \at
    Technische Universiteit Eindhoven, Eindhoven, Netherlands
    \email{g.h.l.fletcher@tue.nl}
    \and
    N. Yakovets (0000-0002-1488-1414) \at
    Technische Universiteit Eindhoven, Eindhoven, Netherlands
    \email{n.yakovets@tue.nl}
}

\date{Received: date / Accepted: date}

\maketitle

\begin{abstract}
\sloppy
We study the optimization of navigational graph queries in the form of the Regular Queries (RQs), i.e., queries which combine recursive and pattern-matching fragments. Current approaches to their evaluation are not effective in practice.
Towards addressing this, we present a number of novel powerful optimization techniques which aim to constrain the intermediate results during query evaluation.
We show how these techniques can be planned effectively and executed efficiently towards the first practical evaluation solution for complex navigational queries on real-world workloads.
Indeed, our experimental results show several orders of magnitude improvement in query evaluation performance over state-of-the-art techniques on a wide range of queries on diverse datasets.
\end{abstract}

\section{Introduction}
\label{sec:intro}
\sloppy
\added[id=R2,comment={D1}]{State-of-the-art database systems are effectively unable to evaluate complex navigational queries in practice and at scale~\cite{wdbench}.
In this work, we aim to overcome this severe limitation by introducing new optimization techniques for navigational queries with the ultimate goal of finally enabling their evaluation in practice and at scale.}

\added[id=R2,comment={D1}]{Complex navigational queries combine pattern-matching- and recursive path queries~\cite{DBLP:journals/vldb/BonifatiMT20}. More specifically, navigational graph queries in the form of the Regular Queries (RQ)~\cite{DBLP:journals/mst/ReutterRV17}, combine the recursive aspect of path-based languages such as the Regular Path Queries (RPQ)~\cite{DBLP:conf/sigmod/CruzMW87} and the pattern-matching aspect of sub-graph isomorphism queries~\cite{DBLP:journals/jacm/CorneilG70}. This combination of features comprises the very core of the Graph Query Language (GQL), which was recently standardized by ISO and is the product of a collaborative effort between industry and academia to define a \textit{lingua franca} for graph databases~\cite{FrancisGGLMMMPR23,DBLP:conf/sigmod/DeutschFGHLLLMM22}. As such, the efficient evaluation of this type of query is of fundamental importance to the future of graph-data management. Some existing graph-query languages such as SPARQL support only conjunctions over RPQs while others, such as Cypher, are being extended towards the GQL standard}\footnote{As of version 5.9 in the form of quantified path patterns, but without support for nested recursion.}. \added[id=R2,comment={D1}]{The functional language Gremlin is able to express the RQs in full due to being Turing-complete~\cite{DBLP:conf/dbpl/Rodriguez15}.}

\added[id=R2,comment={D1}]{Pattern-matching queries ask for all embeddings of a given query graph in a data graph. Path queries ask for all pairs of start- and end-points of paths formed by finite repetitions of a query pattern (the number of repetitions might not be known at query-time).
In other words, the second type asks for the \textit{transitive closure} of a binary relation which may be obtained by projecting the result of any graph query to a pair of variables.} 

Powerful optimization opportunities arise from the interplay of these two paradigms which together form the complex navigational queries arising in practical graph data management applications. These opportunities have not been exploited to their full potential in the state-of-the-art. 
As we show in this paper, systematically realizing these opportunities enables 
significant
(up to several orders of magnitude) performance improvements over the state-of-the-art on a wide range of queries on diverse datasets. 
We next outline our contributions which enable these important performance gains, towards realizing practical graph querying in contemporary graph applications.

\paragraph{Contribution (1): Novel optimization techniques}
The foundation of our optimizations is that they leverage selectivity, as introduced by filter- and join-predicates, from a query graph that contains transitive closures to constrain the evaluation of those transitive closures and thus improve the overall evaluation performance. We introduce novel optimizations of this type that generalize to closures where both variables participate in join-predicates as part of a conjunctive query. Optimizing such closures in this way is beyond state-of-the-art methods. We will refer to constraining transitive closures using filter- and join-predicates as \textit{seeding}.

For example, consider a social network comprised of people and the friend- and colleague relationships between them, and a query that asks for those pairs of people that are both direct or indirect friends and colleagues. Using state-of-the-art methods, both the friends- and colleagues transitive closures must be evaluated fully, and subsequently intersected. A seeded plan which leverages the potential selectivity of a join between friends and colleagues in order to constrain the evaluation of both transitive closures (e.g., the friends closure is evaluated starting only at those people for whom there exists a friend that also participates in the colleague relationship and vice versa for the colleagues closure) is not considered. Additionally, we observe that when more than two such transitive closures are present in a navigational query, seeding can be applied in a \textit{nested} fashion which \textit{stacks the selectivity of multiple closures}. These novel optimizations techniques are explained in more detail in \autoref{sec:interior-seeding}.

\paragraph{Contribution (2): Graph-structured query plans}
We present a unified graph-structured model for the representation of query plans which can incorporate both tree-based and automata-based optimizations. Query plans- and optimization techniques for queries that contain transitive closures have often been (partially) represented using formalism such as automata \cite{DBLP:conf/sigmod/YakovetsGG16,DBLP:conf/icde/ArroyueloHNR22} or ad-hoc procedural (sub)routines \cite{DBLP:conf/sigmod/JachietGGL20}, instead of the more conventional \textit{tree-structured} plans. This complicates, or even excludes, the cross-optimization of recursive- and non-recursive parts of a query and the use of well-known optimization techniques for tree-structured plans. Graph-structured query plans are explained in more detail in \autoref{sec:query_plan_rep}.

\paragraph{Contribution (3): Provably efficient enumeration}
We outline a solution to the enumeration problem using a \textit{fast}, \textit{scalable} and \textit{extendable} enumerator featuring a \textit{top-down} design relying on \textit{memoization} which guarantees the optimality of the chosen plan with respect to the cost model.
Many logical equivalences exist for logical operations in database queries. These equivalences lead to a multitude of semantically equivalent \textit{logical query plans} for any query. The fact that two semantically equivalent plans can display drastically different performance on the same database instance necessitates the existence of a \textit{query optimizer} as part of any database system. The optimizer's task is to pick, for any combination of input query and database instance, an \quotes{optimal} query plan. Therefore, a major sub-task of the optimizer is to \textit{enumerate} a set of semantically equivalent query plans for the input query (enumeration of query plans should not be confused with enumeration of query \textit{results}). The plans in this set can then be \emph{costed} against the database instance in an attempt to pick a good query plan for that instance. Given an input query and a finite set of logical equivalences we wish to consider, we define the \textit{enumeration problem} as the problem of enumerating all semantically equivalent query plans that can be obtained by any finite, acyclic sequence of applications of the logical equivalences.

Solving the enumeration problem for recursive queries over graph-structured plans is not trivial. Enumerating plans with tree-breaking features such as operators with multiple consumers and \quotes{cross-optimizable} fix-point procedures (i.e., not black boxes) requires a degree of flexibility of an enumerator that is difficult to achieve using a conventional design. We design a \textit{top-down} enumerator that relies on \textit{memoization}, embeds instances of recursive sub-problems (i.e., sub-queries to be optimized) directly into partial solutions and which can guarantee the optimality of the chosen plan with respect to a cost model in \autoref{sec:enum}.

Finally, we study the complexity of the enumeration procedure and investigate in particular the complexity added by the proposed optimizations. Most importantly, we prove that adding the proposed optimizations increases the complexity only up to a \textit{constant factor}, which effectively enables the application of these optimizations in practice. The complexity of the enumeration procedure is discussed in more detail in \autoref{sec:complexity}.

\paragraph{Contribution (4): Practical feasibility}
We show how to implement the described model of query plans, top-down enumeration procedure and query optimization techniques related to seeding
in an industrial graph-database engine. Neither the optimization techniques, the query plan model nor the enumerator's design are particular to this engine and all could be applied as part of any other system. 

We demonstrate enumeration of a diverse and hitherto unexplored set of query plans that employ seeding, and present empirical evidence showing that enumerating such plans is feasible in practice and that cost-based optimization can be employed effectively to reduce the time spent on- and amount of data processed during query evaluation by up to three orders of magnitude. Feasibility and performance are discussed in \autoref{sec:experiments}.

\paragraph{Contribution (5): Support for Regular Queries}
Seeding has not been studied in the context of the Regular Queries (RQs) which are strictly more expressive than UCRPQs, are closed under the transitive closure operation, have a decidable query containment problem, are at the heart of industrial-strength languages such as SPARQL and GQL, and are therefore the natural choice of query formalism for navigational queries.
The increased expressivity of RQs was shown to be incredibly useful in many practical application areas~\cite{DBLP:journals/mst/ReutterRV17}.

For example, consider a financial network comprised of people and accounts and the ownership- and transaction relationships between them. A query that asks for all pairs of people that own accounts for which there exists a direct or indirect transaction relationship from one account to the other and that any intermediary accounts have a direct or indirect transaction relationships with the same fixed account, is a query that can be represented as an RQ, but not as a UCRPQ because of the predicate that is imposed upon every intermediary account. Such queries are important for validation- or fraud detection purposes.

In summary, graph-structured query plans provide the infrastructure necessary to express existing- and novel optimization techniques for navigational graph queries. The novel techniques are the application of seeding to RQs, to transitive closures of binary relations where both variables participate in join-predicates as part of a larger query, and the ability to apply seeding in a nested fashion when multiple transitive closures are affected. We show that these novel optimizations are powerful and can be applied effectively.

The structure of this paper is as follows. Property graphs are introduced as the data model and the Regular Queries are defined in \autoref{sec:spec}. The concept of seeding and two novel extensions to it are discussed in \autoref{sec:seeding}. The enumerator is discussed in terms of its design, the way it represents query plans and the rules it uses in \autoref{sec:enum}, alongside the complexity of the enumeration procedure. The experimental results and a case study are presented in \autoref{sec:method}, \autoref{sec:use_case_study}. Finally, the related work and conclusions are discussed in \autoref{sec:related} and \autoref{sec:conclusion}.

\section{Data model and query specification}
\label{sec:spec}

\subsection{Data model}
\label{sec:spec-data}
Without loss of generality, we adopt the data model of property graphs, recently standardized by ISO \cite{DBLP:series/synthesis/2018Bonifati,FrancisGGLMMMPR23}. Let $\mathcal{K}$ be a finite set of property keys, $\mathcal{N}$ a finite set of property values, $I \subset \mathbb{N}$ a set of numerical identifiers, and $I_V$ and $I_E$ finite subsets of $I$ such that $I_V \cap I_E = \emptyset$ called vertex- and edge identifiers, respectively. A property graph is defined as a pair $G = (E, P)$ with $E \subseteq I_V \times I_E \times I_V$ and $P \subseteq (I_V \cup I_E) \times \mathcal{K} \times \mathcal{N}$. That is, $E$ is a set of triples of the form $(s, e, t)$ where $s$ and $t$ are vertex identifiers representing the end-points of a directed edge from $s$ to $t$, while $e$ uniquely identifies an edge, and $P$ is a set of triples of the form $(o, k, v)$ where $o$ is a vertex- or edge identifier, $k$ a property key and $v$ a property value.

\subsection{Regular Queries}
\label{sec:spec-rqs}
The Regular Queries (RQs) extend non-recursive Datalog with a transitive closure operator on binary predicates. 
The RQs are a natural, rich and well-behaved query language for graphs:
a) RQs are closed under the transitive closure operation; b) RQs generalize  and extend existing fundamental graph query languages such as unions of conjunctive 2-way regular path queries (UC2RQPs) \cite{DBLP:conf/kr/CalvaneseGLV00/uc2rpq} and unions of conjunctive nested 2-way regular path queries (UCN2RPQs)~\cite{DBLP:conf/amw/BarceloPR12/uc2nrpq}; and c) the query containment problem is decidable for RQs,
 opening the door to principled and formally well-grounded optimization solutions. The evaluation problem for RQs is \textsc{NP}-complete in combined complexity and \textsc{NLogSpace}-complete in data complexity. The semantics of RQs are inherited from the semantics of Datalog in the natural way~\cite{DBLP:journals/mst/ReutterRV17}.

\subsubsection{Datalog}
Datalog is a declarative, logic-based language in which rule-based \textit{programs} can be specified that derive all \textit{facts} (i.e., tuples) that follow logically from a set of given facts and the rules in the program \cite{DBLP:journals/tkde/CeriGT89,DBLP:journals/ftdb/GreenHLZ13}.

A Datalog \textit{program} is a set of Datalog \textit{rules}. Each rule is an expression of the form:
\begin{equation}
  p(t_1, ..., t_n) \,\leftarrow\, L_1, L_2, ..., L_m \nonumber
    \label{eq:datalog_rule}
\end{equation}
where $p(t_1, ..., t_n)$ is the \textit{head} of the rule with \textit{terms} $t_1, ..., t_n$ and  \textit{literals} $L_i$ for $1 \leq i \leq m$. The conjunction $L_1, L_2, ..., L_m$ is referred to as the \textit{body} of the rule. A term is either a variable or a constant. A \textit{predicate} is the name of a rule or of a given set of facts. An \textit{atom} is a predicate along with a list of terms (i.e., $p(t_1, ..., t_n)$ is an atom with predicate $p$ and list of terms $t_1, ..., t_n$). For our purposes, a literal is always an atom, since the use of \textit{negative} literals and the non-monotonic reasoning that follows from them are beyond the scope of this work.

A predicate is either \textit{extensional} or \textit{intensional}. An extensional predicate refers to a given (at query time) set of facts. In other words, a set of extensional predicates together can be said to capture a database instance, just like a set of \textit{tables} capture a database instance in a relational-database setting. An intensional predicate is any predicate that is not extensional and is used as the head of one or more rules.

The informal semantics of evaluating a rule $p(t_1, ..., t_n) \leftarrow L_1, L_2, ..., L_m$ are to produce a set called $p$ of tuples whose values bind to $t_1, ..., t_n$ by replacing all variables in the body by values for which the conjunction of the literals in the body is true.

\begin{figure*}[ht]
    \centering
    \includegraphics[width=\textwidth]{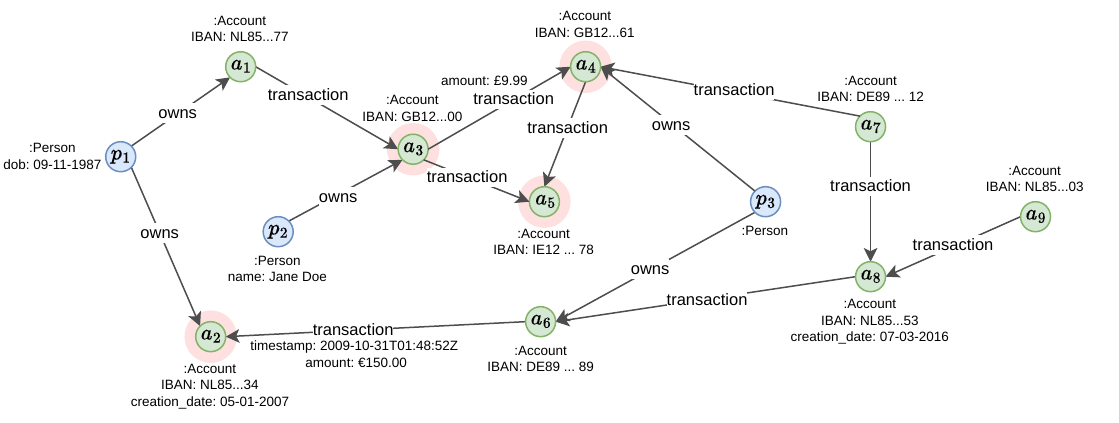}
    \caption{A property graph capturing financial data. Nodes highlighted in red comprise the seed $S$ from \autoref{prog:simple_seeded}.}
    \label{fig:prop_graph_example}
\end{figure*}

\subsubsection{Example Regular Query}
Let $G = (E, P)$ be a property graph with $E$ capturing the \textit{topology} of the graph and $P$ it's \textit{properties}. Throughout the remainder of this paper, we will refer to the extensional predicates $E(s, e, t)$ and $P(o, k, v)$ when defining Regular Queries and Datalog programs. Together these predicates capture the graph $G$.

Consider again the financial network and the natural language query on it from the introduction. An instance of such a network is displayed in \autoref{fig:prop_graph_example}. Vertices representing people and accounts are denoted $p_i$ and $a_j$, respectively. The query that asks for all pairs of people who own accounts for which there exists a path (i.e., sequence) of transaction relationships from one account to the other and that for all accounts on this path, there exists another path of transaction relationships to the account with IBAN \texttt{IE12 BOFI 9000 0112 3456 78}, is captured by \autoref{q:rq_trans} shown below as a Datalog program. Here, $R^+(s, t)$ denotes a transitive closure over binary relation $R(s, t)$.
\begin{qequation}
  \begin{aligned}[b]
    O(s, t) & \leftarrow E(s, e, t), P(e, \text{label}, \text{owns}) \\
    T(s, t) & \leftarrow E(s, e, t), P(e, \text{label}, \text{transaction}) \\
    F(s) & \leftarrow T^+(s, t), P(t, \text{IBAN}, \text{\texttt{IE12...78}}) \\
    I(x, y) & \leftarrow\, T(x, y), F(x) \\
    Ans(w, z) & \leftarrow\, O(w, x), I^+(x, y), O(z, y), F(y)
  \end{aligned}
  \label{q:rq_trans}
\end{qequation}
Because $I$ is defined in terms of a conjunction between $T$ and $F$, and \autoref{q:rq_trans} is defined in terms of $I^+$ this query cannot be captured by a UCN2RPQ, which do not allow for the transitive closure operation to be applied to the result of a conjunction. Hence, this query is an RQ but not a UCN2RPQ. One pair in the result of \autoref{q:rq_trans} on the property graph from \autoref{fig:prop_graph_example} is $(p_1, p_3)$. There does not exist a direct transaction from account $a_1$ to $a_5$ (the account with the desired IBAN), but the definition of $F(s)$ is in terms of $T^+$, not just $T$. A path of transactions exists from $a_1$ to $a_5$ via $a_3$.

\section{Seeding}
\label{sec:seeding}

Let $R(s, t)$ be a binary relation and $R^+(s, t)$ its transitive closure. A naive solution is to compute $R^+(s, t)$ \quotes{in full} by starting from $R(s, t)$ and repeatedly \textit{expanding} by one join with $R(s, t)$ until a fix-point is reached (i.e., no new pairs are added by another join with $R(s, t)$). This strategy is defined by Datalog \autoref{prog:linear_recursion}.

\begin{dequation}
    \begin{aligned}[b]
        R^+(s, t) & \leftarrow R(s, t) \\
        R^+(s, t) & \leftarrow R^+(s, x), R(x, t) \\
    \end{aligned}
    \label{prog:linear_recursion}
\end{dequation}

Whenever $R^+(s, t)$ is also the desired output of the query, computing the transitive closure in full cannot be avoided. However, as soon as the query contains other predicates besides the transitive closure (i.e., the closure is \textit{embedded} into a larger query), computing it in full is often not necessary. Instead, we can \textit{seed} a transitive closure to avoid the exploration of paths that do not join with the larger query into which it is embedded. We define four key concepts involved in seeding.

\begin{definition}
    A \textit{seed} $S$ is a set of nodes from a property graph that is used to restrict the transitive closure of a binary relation over the set of nodes in the graph.
\end{definition}

\begin{definition}
    A \textit{seeding query} $Q_s$ is a query whose result is used to obtain one or more seeds.
\end{definition}

\begin{definition}
    A \textit{seeding relation} $R(x_1,..., x_n)$ is the result of evaluating a seeding query on a property graph and it is an $n$-ary relation where each variable $x_i$ binds to a node in a graph.
\end{definition}

\begin{definition}
    Given a seed $S$, binary relation $T(x, y)$ and $T^+(x, y)$ its transitive closure, a \textit{seeded closure} $\scra{T}{S}$ for $T^+$ on $S$ containing only end-points of paths \textit{starting} at nodes in $S$ is defined as:
    \begin{equation*}
        \scra{T}{S} = \{\, (u, v) \,|\, (u, v) \in T^+ \wedge u \in S \,\} \cup \{\, (u, u) \,|\, u \in S \,\}
    \end{equation*}
    Symmetrically, a seeded closure $\scrp{T}{S}$ containing only end-points of paths \textit{ending} at nodes in $S$ is defined as:
    \begin{equation*}
        \scrp{T}{S} = \{\, (u, v) \,|\, (u, v) \in T^+ \wedge v \in S \,\} \cup \{\, (v, v) \,|\, v \in S \,\}
    \end{equation*}
\end{definition}

That is, $\scra{T}{S}$ is the union between the identity relation of $S$ and the subset of $T^+$ containing only those pairs of nodes $(u, v)$ where $u$ occurs in the seed $S$. The identity relation is included to guarantee that every tuple in the seeding relation $R$ from which $S$ was obtained joins with at least one pair in $\scra{T}{S}$.

Consider a query that asks for all pairs of people and accounts $(p, a_t)$ s.t. $p$ owns an account $a_s$ from which there exists a path of transactions to $a_t$. This query is captured as an RQ by \autoref{q:simple} (re-definitions of $O$ and $T$ omitted), and it is displayed diagrammatically as a query-graph in \autoref{fig:query_2}.
\begin{qequation}
  \begin{aligned}[b]
    Ans(x, z) & \leftarrow O(x, y), T^+(y, z)
  \end{aligned}
  \label{q:simple}
\end{qequation}

\autoref{q:simple} joins predicates $O(x, y)$ and $T^+(y, z)$ on variable $y$. This presents an opportunity to avoid computing $T^+$ in full. The strategy is as follows, we first define a \textit{seeding query} $Q_s$ as the join between $O(x, y)$ and the base of the transitive closure, which is $T(y, z)$. Evaluating this seeding query yields the \textit{seeding relation} $R$. Then, we obtain a seed $S$ by projecting $R$ to $z$ and compute a \textit{seeded closure} $\scra{T}{S}(y, z)$. Finally, we join the seeded closure and the seeding relation.
The set of nodes comprising the seed $S$ for the property graph from \autoref{fig:prop_graph_example} are highlighted in red. Note that more than half of the nodes representing accounts are \textit{not} part of $S$ and so no paths \textit{starting} at any of these nodes will be explored. This is the mechanism by which performance gains are made.\autoref{prog:simple_seeded} defines the seeding strategy as a Datalog program and its application is displayed diagrammatically in \autoref{fig:program_d2}.

\begin{dequation}
    \begin{aligned}[b]
    R(w, y) & \leftarrow O(w, x), T(x, y) \\
    S(y) & \leftarrow R(\_, y) \\
    \scra{T}{S}(y, y) & \leftarrow S(y) \\
    \scra{T}{S}(y, z) & \leftarrow \scra{T}{S}(y, v), T(v, z) \\
    Ans(w, z) & \leftarrow R(w, y), \scra{T}{S}(y, z)
    \end{aligned}
    \label{prog:simple_seeded}
\end{dequation}

The strategy of rewriting a transitive closure in terms of a seed which leverages conjunctions and filters from the query into which the transitive closure is embedded in order to reduce the number of tuples processed in the fix-point procedure is what we will refer to as \textit{seeding}. We distinguish two types of transitive closures in \autoref{sec:closure-types}, and show how one of the types can benefit from seeding using novel methods that are beyond the capabilities of state-of-the-art methods in \autoref{sec:interior-seeding}.

\subsection{Interior- and exterior closures}
\label{sec:closure-types}
In order to go beyond the known strategies for seeding~\cite{DBLP:conf/sigmod/YakovetsGG16,DBLP:conf/sigmod/JachietGGL20}, we distinguish two types of transitive closures, namely \textit{interior} and \textit{exterior} closures. An interior closure is any transitive closure predicate $T^+(s, t)$ where both $s$ and $t$ occur in at least one other (but not necessarily the same) predicate in the query. An exterior closure is any transitive closure predicate $T^+(s, t)$ where either $s$ or $t$, but not both, occur in at least one other predicate in the query. The transitive closure $T^+(y, z)$ from \autoref{q:simple} is an \textit{exterior} closure because $y$ occurs in $O(x, y)$, but $z$ does not occur in any predicate other than $T^+(y, z)$. Conversely, the transitive closure $I^+(x, y)$ from \autoref{q:rq_trans} is an \textit{interior} closure because $x$ occurs in $O(w, x)$ and $y$ occurs in $O(z, y)$.

\subsection{Seeding interior closures}
\label{sec:interior-seeding}
Existing approaches can seed only exterior closures because they require that one variable is \quotes{free} (i.e., does not occur in any other predicate) \cite{DBLP:conf/sigmod/YakovetsGG16,DBLP:conf/sigmod/JachietGGL20}. There may be opportunities to seed interior closures by evaluating the join-predicates on one of its variables only after the transitive closure has been evaluated. In this way, a variable is \quotes{freed} so that the closure can be seeded.
To avoid Cartesian products in query plans, we require that the seeding query's join-graph remains connected.

\begin{qequation}
    \begin{aligned}[b]
    X(s, t) & \leftarrow E(s, e, t), P(e, \text{label}, l_x) \\
    Y(s, t) & \leftarrow E(s, e, t), P(e, \text{label}, l_y) \\
    Z(s, t) & \leftarrow E(s, e, t), P(e, \text{label}, l_z) \\
    Ans(s, t) & \leftarrow X^+(s, t), Y^+(s, t), Z^+(s, t)
    \end{aligned}
    \label{q:pcc3}
\end{qequation}

Consider \autoref{q:pcc3} which asks for pairs of vertices $(s, t)$ such that there exist three paths of edges labelled $l_x, l_y$ and $l_z$ respectively, from $s$ to $t$. The query-graph for this query is displayed in \autoref{fig:query_3}. A possible strategy for this query would be to define the seeding relation $R$ as the conjunction of $X, Y$ and $Z$ \textit{only} on the variable $s$ but not on $t$ (or the other way around) and to seed \textit{each} transitive closure on $R$ projected onto a single (different) variable for $X, Y$ and $Z$. This strategy is defined by \autoref{prog:pcc3_seeded_1} and is displayed in \autoref{fig:program_d3}.

\begin{dequation}
  \begin{aligned}[t]
    & R(s, x, y, z) \leftarrow X(s, x), Y(s, y), Z(s, z) \\
    & S(x) \leftarrow R(\_, x, \_, \_) \\
    & T(y) \leftarrow R(\_, \_, y, \_) \\
    & U(z) \leftarrow R(\_, \_, \_, z) \\
    & \scra{X}{S}(x, x) \leftarrow S(x) \\
    & \scra{X}{S}(s, x) \leftarrow \scra{X}{S}(s, w), X(w, x) \\
    & \scra{Y}{T}(y, y) \leftarrow T(y) \\
    & \scra{Y}{T}(s, y) \leftarrow \scra{Y}{T}(s, w), Y(w, y) \\
    & \scra{Z}{U}(z, z) \leftarrow U(z) \\
    & \scra{Z}{U}(s, z) \leftarrow \scra{Z}{U}(s, w), Z(w, z) \\
    & {\scriptstyle Ans(s, t)} \leftarrow R(s, x, y, z), \scra{X}{S}(x, t), \scra{Y}{T}(y, t), \scra{Z}{U}(z, t)
  \end{aligned}
  \label{prog:pcc3_seeded_1}
\end{dequation}

In this case, seeding the interior closures is possible because the definition of $R(s, x, y, z)$ consists of three predicates $X, Y$ and $Z$ which have a common variable $s$ and so its join graph is connected. The interior closure $I^+(x, y)$ from \autoref{q:rq_trans} \textit{cannot} be seeded because $O(w, x)$ would become disconnected if the variable $x$ is freed, and both $O(z, y)$ and $F(y)$ would become disconnected (from the remainder of the query, not from each other) if the variable $y$ is freed.

\subsubsection{Stacking selectivity}
In \autoref{prog:pcc3_seeded_1}, all three of the transitive closures are expanded on top of the \textit{same} seeding relation. That is, the seeds $S, T$ and $U$ are all projections onto a single variable of the same seeding relation $R$. In this case, there is an opportunity to use a potentially \textit{more restrictive} (i.e., lower cardinality) seed for one of the three transitive closures. Consider the non-recursive rule \mbox{$\scra{Z}{U}(z, z) \leftarrow U(z)$} which populates $\scra{Z}{U}$ with the set of vertices that have an edge labelled $l_z$ incoming on them, such that the source of this edge has edges labelled $l_x$ and $l_y$ outgoing from it. Now consider the set of vertices that have an edge labelled $l_z$ incoming on them, such that the source of this edge has a \textit{path} of edges labelled $l_x$ and a \textit{path} of edges labelled $l_y$ outgoing from it, such that these two paths \textit{end up in the same vertex}. The second set is a subset of the first set, because two outgoing paths imply at least two outgoing edges, but two outgoing edges may not lead to paths that \textit{end up in the same vertex}. Therefore, one possible strategy to reduce the size of the seed and ultimately to further reduce the number of tuples processed in the fix-point procedure for $\scra{Z}{U}$ is to use the intersection of $\scra{X}{S}$ and $\scra{Y}{T}$ as a \textit{new seeding relation} to base the expansion of $\scra{Z}{U}$ on. This strategy \textit{stacks the selectivity} of the fact that the convergence of $i$ paths with $i > 1$ may be used to produce a more restrictive seed for the $i+1^{th}$ transitive closure. This strategy is defined by \autoref{prog:pcc3_seeded_2} (re-definitions of $X$, $Y$ and $Z$ are omitted) and is displayed in \autoref{fig:program_d4}.

\begin{dequation}
  \begin{aligned}[b]
    R(s, x, y, z) & \leftarrow X(s, x), Y(s, y), Z(s, z) \\
    S(x) & \leftarrow R(\_, x, \_, \_) \\
    T(y) & \leftarrow R(\_, \_, y, \_) \\
    \scra{X}{S}(x, x) & \leftarrow S(x) \\
    \scra{X}{S}(s, x) & \leftarrow \scra{X}{S}(s, w), X(w, x) \\
    \scra{Y}{T}(y, y) & \leftarrow T(y) \\
    \scra{Y}{T}(s, y) & \leftarrow \scra{Y}{T}(s, w), Y(w, y) \\
    W(s, w, z) & \leftarrow R(s, x, y, z), \scra{X}{S}(x, w), \scra{Y}{T}(y, w) \\
    U(z) & \leftarrow W(\_, \_, z) \\
    \scra{Z}{U}(z, z) & \leftarrow U(z) \\
    \scra{Z}{U}(s, z) & \leftarrow \scra{Z}{U}(s, v), Z(v, z) \\
    Ans(s, t) & \leftarrow W(s, t, z), \scra{Z}{U}(z, t)
  \end{aligned}
  \label{prog:pcc3_seeded_2}
\end{dequation}

This optimization is very powerful when the expansion of $\scra{Z}{U}$ represents a large proportion of the total amount of work done to evaluate the query, but the intersection between $X^+$ and $Y^+$ is selective. Regarding the applicability of this optimization it is important to note that the interior closures can \quotes{occur anywhere} in the input query and do not need to be adjacent, but \autoref{q:pcc3} makes for a succinct example because all three closures are adjacent.

Note that the choice to expand $\scra{Z}{U}$ \quotes{last} (i.e., on top of a more restrictive seed) was taken only for the sake of the example. Depending on the cardinalities of $X^+, Y^+$ and $Z^+$ it may be preferable to expand one of the other two closures last.
\begin{figure*}[ht]
    \centering
     \begin{subfigure}[t]{0.29\textwidth}
         \includegraphics[width=\linewidth]{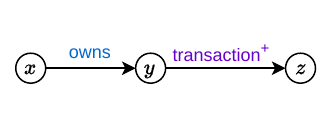}
         \caption{Query-graph for \hyperref[q:simple]{\ref*{q:simple}}}
         \label{fig:query_2}
     \end{subfigure}
     \hfill
     \begin{subfigure}[t]{0.29\textwidth}
         \includegraphics[width=\linewidth]{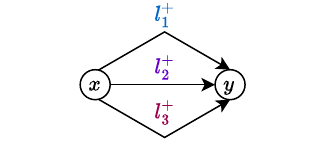}
         \caption{Query-graph for \hyperref[q:pcc3]{\ref*{q:pcc3}} \& \hyperref[eq:template_and_base_definitions]{PCC3}}
         \label{fig:query_3}
     \end{subfigure}
     \hfill
     \begin{subfigure}[t]{0.32\textwidth}
         \includegraphics[width=\linewidth]{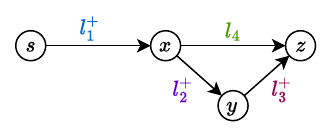}
         \caption{Query-graph for \hyperref[q:rules_ex]{\ref*{q:rules_ex}}}
         \label{fig:query_4}
     \end{subfigure}
     \\
     \begin{subfigure}[t]{0.32\textwidth}
         \includegraphics[width=\linewidth]{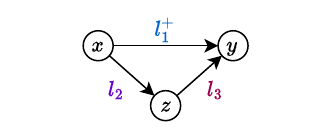}
         \caption{Query-graph for \hyperref[eq:template_and_base_definitions]{CCC1}}
         \label{fig:query_ccc1}
     \end{subfigure}
     \hfill
     \begin{subfigure}[t]{0.32\textwidth}
         \includegraphics[width=\linewidth]{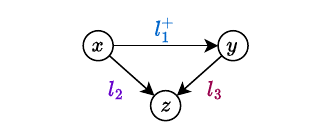}
         \caption{Query-graph for \hyperref[eq:template_and_base_definitions]{CCC2}}
         \label{fig:query_ccc2}
     \end{subfigure}
     \hfill
     \begin{subfigure}[t]{0.32\textwidth}
         \includegraphics[width=\linewidth]{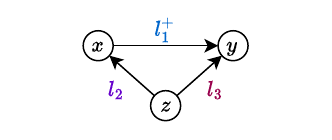}
         \caption{Query-graph for \hyperref[eq:template_and_base_definitions]{CCC3}}
         \label{fig:query_ccc3}
     \end{subfigure}
     \\
     \begin{subfigure}[t]{0.32\textwidth}
         \includegraphics[width=\linewidth]{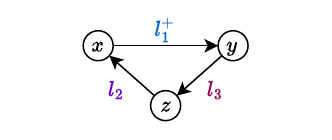}
         \caption{Query-graph for \hyperref[eq:template_and_base_definitions]{CCC4}}
         \label{fig:query_ccc4}
     \end{subfigure}
     \hfill
     \begin{subfigure}[t]{0.32\textwidth}
         \includegraphics[width=\linewidth]{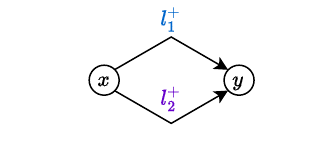}
         \caption{Query-graph for \hyperref[eq:template_and_base_definitions]{PCC2}}
         \label{fig:query_pcc2}
     \end{subfigure}
     \hfill
     \begin{subfigure}[t]{0.32\textwidth}
         \includegraphics[width=\linewidth]{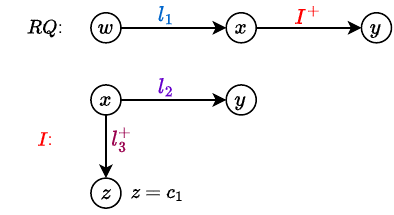}
         \caption{Query-graph for \hyperref[eq:template_and_base_definitions]{RQ}}
         \label{fig:query_rq}
     \end{subfigure}
     \caption{Query-graphs of various queries and templates}
     \label{fig:queries_overview}
\end{figure*}
\begin{figure}[ht]
    \centering
    \includegraphics[width=0.95\columnwidth]{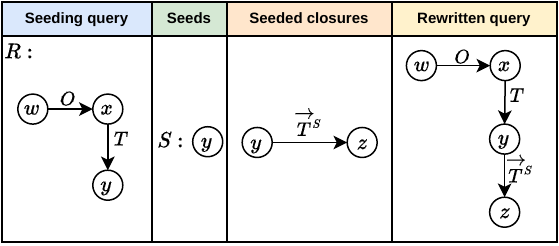}
    \caption{Application of seeding in \autoref{prog:simple_seeded}}
    \label{fig:program_d2}
\end{figure}
\begin{figure}[ht]
    \centering
    \includegraphics[width=0.95\columnwidth]{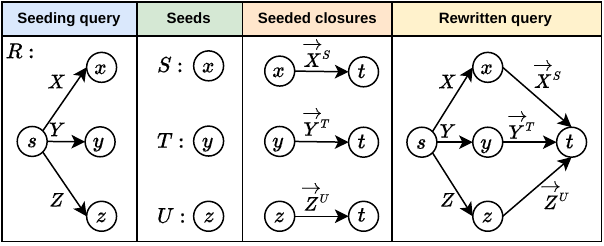}
    \caption{Application of seeding in \autoref{prog:pcc3_seeded_1}}
    \label{fig:program_d3}
\end{figure}
\begin{figure*}[ht]
    \centering
    \includegraphics[width=0.95\textwidth]{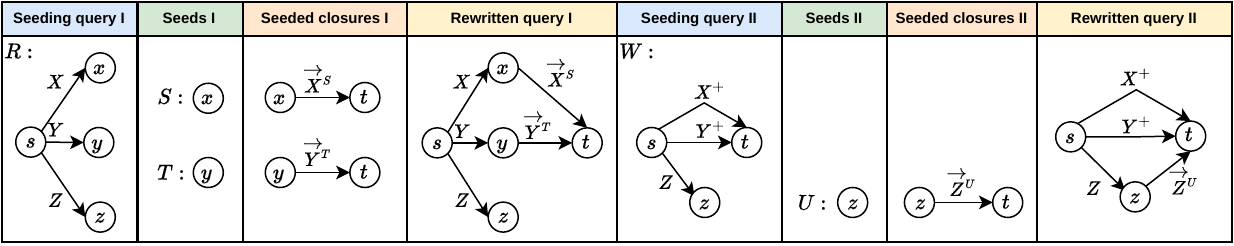}
    \caption{Application of seeding in \autoref{prog:pcc3_seeded_2}}
    \label{fig:program_d4}
\end{figure*}

\section{Enumeration}
\label{sec:enum}

Informally, \textit{enumeration} refers to the process of considering a set of equivalent query plans and the \textit{enumerator} refers to the component of the query optimizer that facilitates this process. We elaborate on our enumerator's design, how we had to go beyond conventional design, and how it represents the query plans it enumerates in \autoref{sec:enum-design}. In \autoref{sec:join-rule} and \autoref{sec:seed-rule}, two of the main \textit{enumeration rules} involved in optimizing navigational graph queries are presented. Finally, the complexity of the enumeration procedure is discussed in \autoref{sec:complexity}.

\subsection{Enumerator design}
\label{sec:enum-design}

A conventional way of enumerating (tree-structured) query plans is to use a \textit{bottom-up dynamic programming algorithm} which constructs query plans for increasingly larger sub-expressions of the input query by combining query plans for smaller sub-expressions which have already been constructed \cite{DBLP:conf/sigmod/MoerkotteDP}. Let $n$ be the size of an input query. While a query's size is typically directly proportional to the number of predicates it contains, we refrain from giving an exact definition here. A bottom-up enumerator would start by producing plans for sub-problems of size $1$, which would typically be reading \quotes{base} relations. Then, the algorithm would consider sub-problems of size $i$ with $2 \leq i \leq n$
in increasing order, and for each $i$ consider combining solutions to sub-problems of sizes $k$ and $i - k$ with $1 \leq k \leq i - 1$ into a solution to the sub-problem of size $i$.

A drawback of the bottom-up approach, is that many combinations of solutions to sub-problems of sizes $k$ and $i - k$ will not produce a semantically valid solution to the problem of size $i$. While checking the validity of such a combination may not be very expensive, it is a source of inefficiency nevertheless \cite{DBLP:conf/icde/FenderM11,DBLP:phd/de/Neumann2005a}. This problem is further compounded when we want to consider solutions to a sub-problem of size $i$ that consist of more than two solutions to smaller sub-problems.
Our top-down approach avoids this problem by embedding instances of sub-problems into partial solutions. Therefore, whenever a sub-problem is considered, solving it must produce a semantically valid solution.

Another convention is to have a \textit{rule-based} design. Such a design offers modularity and extensibility to the enumerator. The enumerator's overall design and the set of rules it uses together determine the set of query plans that will be considered for an input query. To extend the enumerator's capabilities, new query features can be easily supported by new rules, and new rules can easily be added to the set of rules.

We adopt the rule-based design for our enumerator. A set of enumeration rules typically contains at least one rule for each algebraic type (e.g., join, select, etc.) supported by the input query language. This ensures that a query plan can be found for any input query, regardless of the nesting of algebraic types. A rule captures the logic of plan construction. It takes as input an embedding of a sub-query in plan and outputs a set of zero or more plans. If a rule does not apply, it can be said to return the empty-set. \autoref{fig:rule_intro} shows the interface of enumeration rules diagrammatically. The use of the abstraction operator $\square(Q)$ is clarified in \autoref{sec:query_plan_rep}.
\begin{figure}[ht]
    \centering
    \includegraphics[width=\linewidth]{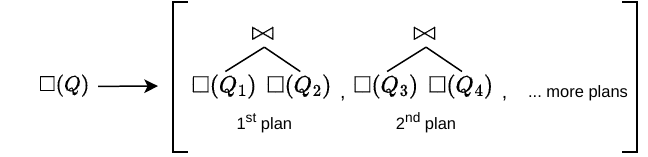}
    \caption{Interface of enumeration rules. $\square(Q)$ denotes an \textit{abstraction} over a query $Q$ (\autoref{sec:query_plan_rep})}
    \label{fig:rule_intro}
\end{figure}

In this paper, we refrain from giving an algebraic specification of the Regular Queries and instead continue to present queries, optimizations and enumeration rules using Datalog which is more convenient for a discussion on seeding. We do note that such algebras exist, for example the Regular Property Graph Algebra (RPGA) \cite{DBLP:series/synthesis/2018Bonifati}. In order to support the other constructs that comprise RQs besides those discussed in the subsequent sections, additional enumeration rules are required. Because these rules are not directly relevant to the topic of seeding, they are presented in an online appendix \cite{Ymous_Optimizing_Navigational_Queries} rather than the main body of this text.

\subsubsection{Query plan representation}
\label{sec:query_plan_rep}
Our model for representing query plans is a directed graph of logical operators.
A \textit{query plan} is an expression of the form $\mathcal{P} = (O, r)$ with $O$ a set of operators and $r \in O$ the root operator. The topology of a plan is captured by a recursive parent-child relationship.
Most operator types are defined in terms of one or more child operators. We define eleven types of operators to be used in query plans. Six of these operator types are well-known from the Relational Algebra, although they might differ slightly from Codd's original specification \cite{codd_ra}. Let $G = (V, E)$ be a property graph and $\mathbb{V}$ a countably infinite set of variable names. We informally define a \textit{join-predicate} as any expression that equates two variables (e.g., $x = y$ with $x, y \in \mathbb{V}$) and a \textit{filter-predicate} as any expression that equates a variable and a constant (e.g., $x = c$ with $x \in \mathbb{V}$). For the sake of simplicity, we will only consider equality predicates and integer constants (i.e., for any constant $c$ we have $c \in \mathbb{N}$), but note that all definitions extends naturally to other types of predicates and constants. We write $\mathbb{J}$ and $\mathbb{F}$ to denote the countably infinite sets of join- and filter-predicates, respectively. The eleven types are defined with respect to $G, \mathbb{V}, \mathbb{J}$ and $\mathbb{F}$ as follows:

\begin{itemize}
    \item $E(i)$ reads the set of edges $E$ and assigns a fresh identifier $i \in \mathbb{N}$ to distinguish multiple references to $E$ in the same query,
    \item $P(i)$ reads the set of properties $P$ and assigns a fresh identifier $i \in \mathbb{N}$ to distinguish multiple references to $P$ in the same query,
    \item $\bowtie(J, c_l, c_r)$ obtains an intermediate result from its left- and right child operators $c_l$ and $c_r$ and joins them based on a set of join-predicates $J \subset \mathbb{J}$,
    \item $\Pi(V, c)$ obtains an intermediate result from its child operator $c$ and projects each tuple only to the variables in the set of variables $V \subset \mathbb{V}$,
    \item $\rho(f, c)$ obtains an intermediate result from its child operator $c$ and renames variables according to a function $f$,
    \item $\sigma(F, c)$ obtains an intermediate result from its child operator $c$ and filters it based on the satisfaction of the conjunction of predicates in the set of filter-predicates $F \subset \mathbb{F}$,
    \item $\cup(c_1, ..., c_n)$ obtains intermediate results from its $n$ child operators and produces their union,
    \item $\alpha(b, c)$ obtains an intermediate result from its child operator $c$ and writes it to a buffer $b$,
    \item $\beta(b)$ obtains an intermediate result by reading from a buffer $b$,
    \item $\delta(c)$ obtains an intermediate result from its child operator $c$ and keeps only those tuples that were never part of this- or any previous intermediate result obtained from $c$, and
    \item $\square(Q)$ encapsulates a query $Q$ inside an \textit{abstraction} operator.
\end{itemize}
The latter four operator types require clarification. The $\alpha$ and $\beta$ operators interact with a buffer, which can be regarded as a temporary materialized view that exists only for the duration of query evaluation. The purpose of such buffers is to go beyond tree-structured query plans and allow operators with more than one consumer. Furthermore, the use of buffers allows a \textit{cyclic flow of tuples} through a query plan, thus enabling plans that can evaluate recursive queries. Note that it is the flow of tuples which may follow a cycle. A query plan is ultimately a high-level representation of a computer program and so cannot, and will not, have \textit{direct} cyclic dependencies between its operators. Such dependencies are avoided by the use of buffers as intermediaries. The $\alpha$ operator captures \textit{writing} to a buffer. As a consequence, there should be exactly one $\alpha$ operator associated with each buffer used by any query plan. The $\beta$ operator captures \textit{reading from} a buffer, and so there should be one or more $\beta$ operators associated with each buffer.
\added[id=R1,comment={D1}]{The use of buffers in query plans does \textit{not} imply that intermediate results with multiple- or cyclic consumers will be (fully) materialized during query evaluation. While materialization offers one way of executing query plans that feature operators with multiple consumers, there are alternative execution strategies that mitigate or remove the need for materalization~\cite{DBLP:phd/de/Neumann2005a} (e.g., \textit{push-based} strategies~\cite{quickstep}) which may be extended to cover plans the feature cyclic consumsers.}

The $\delta$ operator performs two functions related to deduplication of intermediate results. Firstly, it removes duplicates that occur \textit{in the same} intermediate result. Secondly, it removes any tuples from an intermediate result that were part of a previous intermediate result obtained from the $\delta$ operator's child. This is necessary in plans that evaluate closures, in order to reach a fix-point in the presence of cycles in the data graph. If the $\delta$ operator does \textit{not} lie on a cyclic path in the plan graph, the second function is void.

The \textit{abstraction} operator $\square$ allows for a query to be \textit{embedded into a query plan}. Embedding queries into query plans can be thought of as a specific form of embedding instances of recursive sub-problems into partial solutions. This embedding enables the construction of an extremely flexible \textit{top-down} enumerator that relies on memoization, rather than a \textit{bottom-up} enumerator that relies on a more rigid dynamic-programming table.

Because a bottom-up enumerator processes problems in order of increasing size, it is guaranteed to consider, for all problems it processes, all combinations of solutions to sub-problems. Our top-down enumerator must have the same property, which we provide by enforcing a depth-first order on the processing of abstractions. The depth-first order ensures that the size of sub-problems being processed is always decreasing, until a further decrease becomes impossible and the procedure back-tracks. Because multiple (partial) plans can exist simultaneously, and because each plan can contain multiple abstractions, we use a \textit{global} stack of query plans and a stack of abstractions \textit{per plan}. The global stack will be denoted $\mathcal{S}$ and $\mathcal{P.A}$ will denote the stack associated with plan $\mathcal{P}$.

\subsubsection{Enumeration algorithm}

Let $R$ be a finite set of enumeration rules and $\mathcal{S}$ a \textit{stack} of query plans.
The enumeration algorithm then boils down to a \texttt{while}-loop that tries to empty $\mathcal{S}$ by repeatedly applying memoization and the enumeration rules. Pseudo-code for this algorithm is shown in \autoref{alg:enum}.

The set of rules $R$ is fixed at query time and so is not an input parameter to \autoref{alg:enum}. In practice, rules are enabled or disabled by specialized commands executed prior to issuing a query. For the purpose of this paper, the set $R$ contains a rule related to reading triples from the sets $E$ and $P$ comprising a property graph, applying filter-predicates, constructing a fix-point procedure for \textit{unseeded} transitive closures, performing joins, applying seeding to one or more transitive closures and performing logical union of results. We elaborate on the join- and seeding rules in \autoref{sec:join-rule} and \autoref{sec:seed-rule}, respectively. We omit a detailed discussion of the other rules because they are straight-forward and do not contribute directly to the application of seeding.
\begin{algorithm}[ht]
    \caption{Enumerate($Q_{in}$)}
    \label{alg:enum}
    \resizebox{\linewidth}{!}{
    \begin{minipage}{1\linewidth}
    \begin{algorithmic}[1]
        \State $\mathcal{P}_{in} \leftarrow (\{o\}, o = \square(Q_{in}))$
        \State Add $\mathcal{P}_{in}$ to $\mathcal{S}$
        \While{$\mathcal{S}$ is not empty}
            \State Let $\mathcal{P}$ be the top query plan on $\mathcal{S}$
            \If{$\mathcal{P}$ contains no abstractions}
                \State Memoize $\mathcal{P}$, remove it from $\mathcal{S}$ and  \textbf{continue}
            \EndIf
            \State Let $\square(Q)$ be the top abstraction of $\mathcal{P.A}$
            \State Let $M$ be the set of memoized plans for $\square(Q)$
            \If{$M \neq \emptyset$}
                \State Let $\mathcal{P}_M$ be the memoized plan for $\square(Q)$
                \State Let $\mathcal{P}_R$ be a copy of $\mathcal{P}$, replacing $\square(Q)$ by $\mathcal{P}_M$
                \State Add $\mathcal{P}_R$ to $\mathcal{S}$
            \Else
                \For{rule $r$ in rule set $R$}
                    \For{each $\mathcal{P}_R$ generated by applying $r$ to $\square(Q)$}
                        \State Add $\mathcal{P}_R$ to $\mathcal{S}$
                    \EndFor
                \EndFor
            \EndIf
        \EndWhile
        \State \textbf{return} the memoized solution for $\square(Q_{in})$
    \end{algorithmic}
    \end{minipage}
    }
\end{algorithm}

On line 1 of \autoref{alg:enum} an initial plan $\mathcal{P}_{in}$ is constructed that consists of a single operator namely, an abstraction over the input query $Q_{in}$. This plan is subsequently pushed onto the stack $\mathcal{S}$. Repeatedly taking the top plan from $\mathcal{S}$ until it becomes empty results in a depth-first traversal of a tree of query plans as displayed in \autoref{fig:enum_proc}. Lines 5 and 6 identify cases where $\mathcal{P}$ is a finished plan, it is added to the memoization table and the algorithm continues to the next iteration over $\mathcal{S}$. These cases are the leaf nodes highlighted using dashed blue lines \autoref{fig:enum_proc}. When line 7 is reached, the plan $\mathcal{P}$ contains at least one abstraction $\square(Q)$ which must be processed either by replacing it with a plan from the memoization table in lines 10 through 12, or by applying one or more of the enumeration rules in lines 14 through 16. Replacement by a memoized plan is applied in step 10 in \autoref{fig:enum_proc} for example, where the plan for $R$ memoized in step 5 is used. Finally, at line 17 of \autoref{alg:enum}, it must be the case that a plan has been memoized for $Q_{in}$ by virtue of $\mathcal{S}$ being empty. Probing the memozation table for $Q_{in}$ yields a plan that is optimal with respect to the cost model used.

\begin{figure}[ht]
    \centering
    \includegraphics[width=0.95\columnwidth]{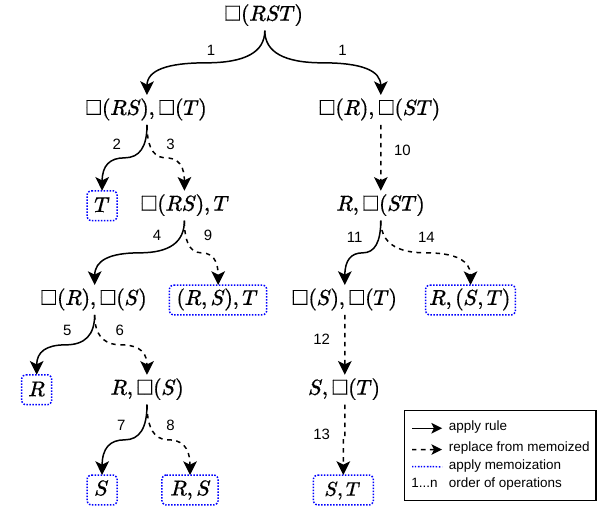}
    \caption{Example enumeration process (\autoref{sec:join-rule})}
    \label{fig:enum_proc}
\end{figure}

\subsection{Join rule}
\label{sec:join-rule}
The join rule can be applied to any query of the form \mbox{$Q(\overline{x}) \leftarrow L_1(\overline{y_1}),..., L_n(\overline{y_n})$} such that the \textit{join-graph} for $Q$ is connected. We write $\overline{x}$ to denote a sequence of variables $x_1,..., x_k$ as is conventional in the context of Datalog, and sometimes abuse notation by interpreting $\overline{x}$ as a set. For a detailed discussion on join-graphs corresponding to conjunctive queries we refer to~\cite{DBLP:conf/icde/FenderM11}. Its purpose is to take the set of predicates in the body $\mathcal{B} = \{L_1, ..., L_n \}$ and construct all possible pairs $(T, U)$ where $T$ and $U$ are strict, non-empty and disjoint subsets of $\mathcal{B}$ such that $T \cup U = \mathcal{B}$ holds, the join-graphs of $T$ and $U$ are connected and there exists at least one join predicate between $T$ and $U$. Note that for any such pair $(T, U)$ there exists a symmetric pair $(U, T)$. We want to consider exactly one pair per symmetric case. The problem of enumerating all such pairs efficiently while avoiding symmetric cases, has been studied extensively in the context of relational databases \cite{ono90,DBLP:conf/icde/FenderM11,DBLP:journals/tkde/FenderM12/reassessing,DBLP:journals/pvldb/FenderM13/counterstrike,DBLP:conf/sigmod/DeHaanT07,DBLP:conf/sigmod/MoerkotteDP}. Hence, we employ a state-of-the-art algorithm developed in the relational context for precisely this purpose, namely \texttt{MinCutBranch} (MCB)\cite{DBLP:conf/icde/FenderM11}.

\begin{qequation}
    \begin{aligned}[t]
        V(s, t) \leftarrow & \, E(s, e, t), P(e, \text{label}, l_1) \label{q:rules_ex} \\
        W(s, t) \leftarrow & \, E(s, e, t), P(e, \text{label}, l_2) \\
        Y(s, t) \leftarrow & \, E(s, e, t), P(e, \text{label}, l_3) \\
        Z(s, t) \leftarrow & \, E(s, e, t), P(e, \text{label}, l_4) \\
        {\scriptstyle Ans(x, y, z)} \leftarrow & \, V^+(s, x), W^+(x, y), Y^+(y, z), Z(x, z)
    \end{aligned}
\end{qequation}

The join rule would construct one plan for each of the following pairs of subsets of $\{V^+, W^+, Y^+, Z\}$ from \autoref{q:rules_ex} whose query-graph is displayed in \autoref{fig:query_4}:

\begin{equation*}
    \begin{aligned}
        (\{V^+\}, \{W^+, Y^+, Z\}) && (\{W^+\}, \{V^+, Y^+, Z\}) \\
        (\{Y^+\}, \{V^+, W^+, Z\}) && (\{Z\}, \{V^+, W^+, Y^+\}) \\
        (\{V^+, W^+\}, \{Y^+, Z\}) && (\{V^+, Z\}, \{W^+, Y^+\})
    \end{aligned}
\end{equation*}

By repeated application of the join rule, all join orders avoiding Cartesian products are constructed. In practice, we typically construct more than one plan per pair of subsets. First, the pair is \quotes{switched} to its symmetric counter-part thus capturing the commutativity of the join operation. Furthermore, various types of annotations may be added to distinguish hash-, merge- and nested-loop joins, etc. Provided that the set of equivalences (e.g., commutativity) and join algorithms (e.g., hash join) is fixed at runtime, the number of plans constructed per pair of subsets is constant. In order to simplify any analysis of the join rule, we will assume that only one plan is constructed per pair of subsets.

As a consequence, the \textit{search space} considered (i.e., the set of join orders constructed) grows exponentially in terms of the size of the input query in the worst case, where the query's join-graph is a \textit{clique}~\cite{DBLP:conf/icde/FenderM11}. The use of \textit{heuristics} has been proposed to reduce this search space and thus speed up query optimization. One such heuristic is to avoid the construction of \quotes{bushy} plans in favor of \quotes{zig-zag} plans~\cite{DBLP:journals/vldb/LeisRGMBKN18}. A bushy plan features at least one join whose inputs are themselves the (indirect) result of one or more joins. Conversely, a zig-zag plan guarantees that for every join in the plan, at least one of the inputs is not the (indirect) result of any join. In \autoref{sec:complexity} we will see that the application of such heuristics combines well with the application of seeding, as a consequence of \autoref{theo:enum}. Because the number of plans generated when seeding is applied is within a \textit{constant} factor of the number of plans generated when seeding is \textit{not} applied, any reduction in terms of the search space's size by the application of heuristics anywhere in the optimization process is, up to a constant factor.

The enumeration process for a query that expresses a join between relations $R$ and $S$ and $S$ and $T$ (but not directly between $R$ and $T$) is displayed in \autoref{fig:enum_proc}. This process can be thought of as walking a tree where the root is an abstraction over the input query (in this case represented as $\square(RST)$) and the leaves are fully fledged or \textit{concrete} logical plans. The inner vertices of this tree are \textit{partial} plans. Step 1 is to take $\square(RST)$ and apply the set of enumeration rules. Application of the join rule creates two plans which each contain a pair of abstractions. Step 2 is to take the top plan and its top abstraction (in this case $\square(T)$) and apply the enumeration rules again. This yields one concrete plan, which is subsequently \textit{memoized}. Step 3 replaces $\square(T)$ in the plan produced in step 1 with the plan for $T$ memoized in step 2. The process of rule application, memoization and replacement repeats until no partial plans remain.

\subsection{Seeding rule}
\label{sec:seed-rule}
The seeding rule takes a conjunctive query that contains at least one closure, tries to construct a seeding query and to obtain seeds for all transitive closures in the query. A more precise description follows in \autoref{sec:seeding-rule-overview} and \autoref{sec:seeding-rule-partitioning}. Construction of the seeding query and the seeded closures are discussed in \autoref{sec:seeding-rule-seeding} through \autoref{sec:seeding-rule-exterior}.

\subsubsection{Seeding rule overview}
\label{sec:seeding-rule-overview}
The seeding rule can be applied to any query of the form \mbox{$Q(\overline{x}) \leftarrow L_1(\overline{y}_1),..., L_n(\overline{y}_n)$} such that the following statements hold:

\begin{itemize}
    \item the join graph for $Q$ is connected,
    \item the body of $Q$ contains at least one transitive closure, and
    \item the body of $Q$ does not contain an interior closure for which freeing either variable produces a disconnected seed.
\end{itemize}

We define the set of predicates in the body of $Q$ as $\mathcal{B} = \{ L_1, ..., L_n \}$ and the union of the output projections of those predicates as $Y = \bigcup_{1 \leq i \leq n}\overline{y}_i$. Note that we abuse notion slightly here, turning sequences of variables $\overline{y}_i$ into a set of variables $Y$. The purpose of the seeding rule is to take $Q$ and \quotes{split} it into a seeding query and $k$ fix-point procedures with $1 \leq k \leq n$ and construct the corresponding query plan. The process can be subdivided into four steps, namely:

\begin{enumerate}
    \item partitioning $\mathcal{B}$ into sets $N, I$ and $X$,
    \item constructing a seeding query based on $N$, $I$ and $X$,
    \item processing each closure in $I$ and,
    \item processing each closure in $X$.
\end{enumerate}

The sets $N, I$ and $X$ are the subsets of $\mathcal{B}$ that contain only the non-recursive predicates, predicates that are interior closures and predicates that are exterior closures, respectively.

\subsubsection{Seeding rule heuristics}
Note that \textit{exactly one} plan is constructed by the seeding rule for any valid input query, even though there are two degrees of freedom related to the set of interior closures $I$ that allow multiple plans to be constructed. Firstly, for any interior closure in $I$ there is the choice of which of the two variables to free. Secondly, there are many ways of stacking the closures in $I$ when $|I| > 2$. In order to construct only a single plan, two \textit{heuristics} will be used to make these two choices.

\added[id=R1,comment={D2}]{We will refer to these two heuristics as $h_1$ and $h_2$. While both heuristics introduce a potential for missing a lower cost plan, they eliminate two sources of exponential increase in the plan space size which render exhaustively enumerating this space infeasible, namely:}
\begin{itemize}
    \item $h1$ avoids enumerating $2^m$ options of which variable to free for $0 < m$ interior closures, and
    \item $h2$ avoids enumerating $\frac{1}{2}m!$ orders of stacking $2 < m$ interior closures.
\end{itemize}

Let $L_i^+$ denote a transitive closure predicate from the body $\mathcal{B}$ of $Q$. We write $h_1(L_i^+)$ to denote the application of $h_1$ to $L_i^+$. The heuristic $h_1$ therefore takes in interior closure $L_i^+$, decides which variable to free, and produces a version of the base $L_i$ in which the freed variable is renamed to a fresh variable (i.e., the join by way of a common variable between it and any other predicate in $\mathcal{B}$ is removed). In our implementation, $h_1$ picks $x$ for any $L_i^+(x, y)$ if it produces a connected seeding relation, and only picks $y$ if it does not.

The heuristic $h_2$ is used to decide the order in which interior closures are stacked (and therefore, the order in which they are joined with their seeding relations). In our implementation, we consider closures in $I$ in order of increasing cardinality estimates. The rationale for this is twofold. First, joining the results of closures in order of increasing cardinality is expected to minimize the output cardinalities of those joins. Second, we in turn expect to minimize the cardinalities of the seeds produced by selectivity stacking.

\added[id=R1,comment={D2}]{Both heuristics are simple, and are to be considered as a baseline effort to render the application of these advanced seeding-based optimization techniques feasible without introducing new sources of exponentional complexity to the optimization process. Considering more advanced and sophisticated heuristics to replace $h_1$ and $h_2$ remains an interesting future challenge.}

\subsubsection{Partitioning closures}
\label{sec:seeding-rule-partitioning}
We partition $\mathcal{B}$ into a set of non-recursive predicates $N$, a set of interior closures $I$ and a set of exterior closures $X$, by checking for each $L_i(\overline{y}_i) \in \mathcal{B}$ whether it is recursive at all, and if so, whether all variables in $\overline{y}_i$ participate in at least one join-predicate. For \autoref{q:rules_ex}, the resulting partitions would be:
\begin{equation*}
  \begin{aligned}[b]
    N = &\{Z(x, z)\}, I = \{W^+(x, y), Y^+(y, z)\} \text{ and} \\
    X = &\{V^+(s, x)\}
  \end{aligned}
  \label{eq:ex_partitioning}
\end{equation*}
This is because $Z(x, z)$ is not recursive, $V^+(s, x)$ is recursive but $s$ does not participate in any join-predicate in \autoref{q:rules_ex}, and $W^+(x, y)$ and $Y^+(y, z)$ are both recursive and $x, y$ and $z$ all participate in at least one join-predicate in \autoref{q:rules_ex}.

\subsubsection{Seeding query construction}
\label{sec:seeding-rule-seeding}
Let $Q_s(\overline{x}_s) \leftarrow \mathcal{B}_s$ be the seeding query. The body $\mathcal{B}_s$ is defined as the conjunction of the predicates in the set:
\begin{equation*}
  \begin{aligned}[b]
    N \cup \{ h_1(L_i^+) \,|\, L_i^+ \in I \} \cup \{ L_i \,|\, L_i^+ \in X \}
  \end{aligned}
  \label{eq:seed_body}
\end{equation*}
For any $L_i \in \mathcal{B}_s$ such that $L_i^+ \in I \cup X$ we may need to extend the set of output variables $\overline{x}_s$. That is, by default $\overline{x}_s = \overline{x}$ but $\overline{x}$ may not contain the free(d) variable from $L_i$, which must be added. For \autoref{q:rules_ex} the resulting seeding query, assuming $h_1$ decides to free $y$ for both $W^+(x, y)$ and $Y^+(y, z)$ and rename it to $y_1$ and $y_2$ respectively, would be:

\begin{equation*}
  \begin{aligned}[b]
    Q_s(s, x, y_1, y_2, z) \leftarrow V(s, x), W(x, y_1), Y(y_2, z), Z(x, z)
  \end{aligned}
  \label{eq:ex_seed}
\end{equation*}

Notice that the variable $s$ has been added to $\overline{x}_s$. This is necessary because the set of nodes that will bind to $s$ during query evaluation, is the set that will be used as the seed for $\scrp{V}{S}(s, x)$.

\subsubsection{Process interior closures}
\label{sec:seeding-rule-interior}
Having defined a query $Q_s$ that captures a seeding query for our query plan, we wish to process each $L_i^+ \in I$ and extend our query plan by adding a fix-point procedure to it (by way of constructing a new sub-graph). Step 1 in \autoref{fig:seed_plan} shows the query plan containing an abstraction over the seeding query $Q_s$ and an $\alpha$-operator writing the evaluation of the seeding query into a buffer for consumption by multiple consumers (i.e., multiple closures from $I \cup X$).
Here we employ heuristic $h_2$ to decide the order in which interior closures are stacked.
After every join that joins the $i^{th}$ transitive closure with the preceding $j$ closures where $1 \leq j < i$, we instantiate a new buffer and $\alpha$-operator that facilitates the selectivity stacking.

Step 2 in \autoref{fig:seed_plan} shows the query plan after processing $I$. The intuition behind this plan is as follows. Query $Q_s$ will be evaluated in some way (as represented by the abstraction), its result will be buffered in $b_1$ and used as a seeding relation to two fix-point produces rooted at the $\alpha(b_2)$ and $\alpha(b_3)$ operators. Both fix-point procedures obtain the result through a $\beta(b_1)$ operator, project the seeding relation to a single (different) variable, thus obtaining their respective seeds. Evaluation of the seeded closures is achieved through the cyclic flow of tuples over the buffers $b_2$ and $b_3$.

\subsubsection{Process exterior closures}
\label{sec:seeding-rule-exterior}
Having processed all interior closures, we wish to process each $L_i^+ \in X$ and extend our query plan by adding a fix-point procedure. \autoref{fig:seed_plan} shows the query plan for $Q$ after processing $X$. Buffer $b_4$ facilitates selectivity stacking. That is, the exterior closure $\scrp{V}{S}$ is computed based on the content of buffer $b_4$ which is the evaluation of the query:

\begin{equation*}
  \begin{aligned}[b]
    Q_{s'}(s, x, y, z) \leftarrow V(s, x), W^+(x, y), Y^+(y, z), Z(x, z)
  \end{aligned}
  \label{eq:ex_seed_2}
\end{equation*}
The query $Q_{s'}$ is a seeding query that is at least as selective as (but in practice often much more selective than) the seeding query $Q_s$ whose evaluation defines the content of buffer $b_1$.

\begin{figure}[t!]
    \includegraphics[width=\columnwidth]{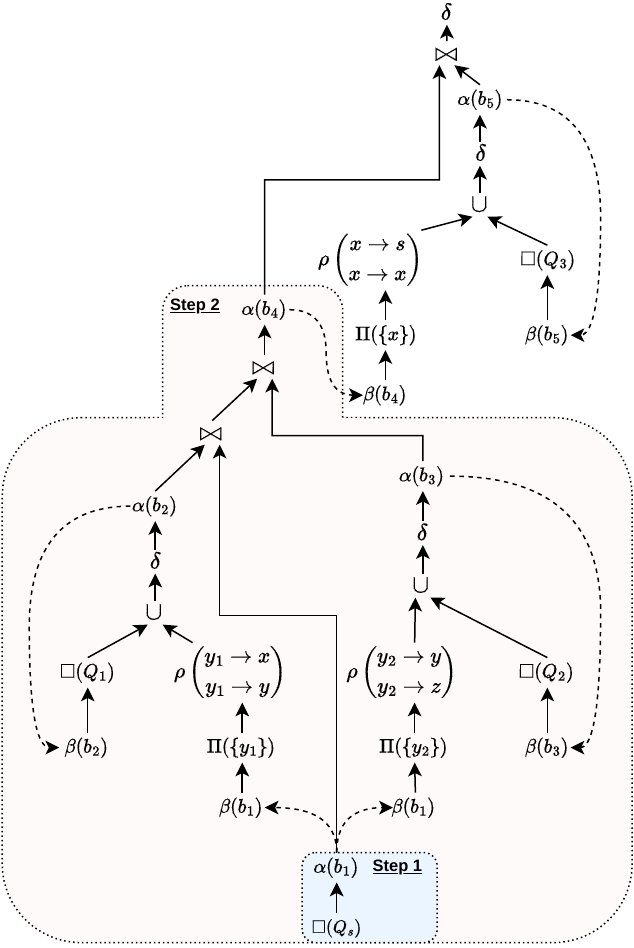}
    \caption{Constructing a seeded plan (\autoref{sec:seeding-rule-interior})}
    \label{fig:seed_plan}
\end{figure}

Notice the inclusion of abstractions over the queries $Q_1, Q_2$ and $Q_3$ in \autoref{fig:seed_plan}. These abstractions capture the fact that the recursive sub-problems of how to expand $\scra{W}{T}, \scrp{Y}{T}$ and $\scrp{V}{S}$ on top of their respective seeds still need to be solved by the enumerator. These queries correspond to the recursive Datalog rules $\scra{W}{T}(x, y) \leftarrow \scra{W}{T}(x, w), W(w, y)$ for $Q_1$, $\scrp{Y}{T}(y, z) \leftarrow Y(y, w), \scrp{Y}{T}(w, z)$ for $Q_2$ and $\scrp{V}{S}(s, x) \leftarrow V(s, w), \scrp{V}{S}(w, x)$ for $Q_3$.

\subsection{Complexity analysis}
\label{sec:complexity}
We analyze the impact that the application of the proposed optimizations has on the complexity of the enumeration procedure. To this end, we will count the \textit{total} number of plans generated during optimization for a query of a particular \textit{size} and \textit{shape}. The following definitions allow us to analyze the complexity of the enumeration procedure.
\begin{definition}
    Let $Q$ be a query. The \textit{optimization tree} $\mathcal{T}_Q$ is the tree of plans generated \textit{implicitly} by applying \autoref{alg:enum} to $Q$.
\end{definition}
A node is added to $\mathcal{T}_Q$ when a plan $\mathcal{P}$ is added to the stack $\mathcal{S}$ in either line 2, 12 or 16 of \autoref{alg:enum}. An edge is added from any plan $\mathcal{P}$ taken from $\mathcal{S}$ in line 4 to any plan $\mathcal{P}_R$ added to $\mathcal{S}$ in either line 12 or 16. The tree is generated implicitly, without being materialized.
\begin{definition}
    Given an optimization tree $\mathcal{T}_Q$, the set of \textit{leaves} of $\mathcal{T}_Q$ is denoted $L(\mathcal{T}_Q)$.
\end{definition}
The set of leaves of an optimization tree $\mathcal{T}_Q$ corresponds to the set of \textit{concrete} plans generated by \autoref{alg:enum} and therefore also to the number of times a cost-model is called when memoizing a plan.
\begin{definition}
    Given a query $Q$, the optimization tree $\mathcal{T}_{Q,u}$ is the tree generated by \autoref{alg:enum} when \textit{not} using the seeding rule.
\end{definition}
\begin{definition}
    Given a query $Q$, the optimization tree $\mathcal{T}_{Q,o}$ is the tree generated by \autoref{alg:enum} when using the seeding rule.
\end{definition}
We derive functions $P_u(n)$ and $P_o(n)$ which count $L(\mathcal{T}_{Q,u})$ and $L(\mathcal{T}_{Q,o})$ respectively, for a query $Q$ of a fixed shape and variable size $n$. Then, we argue that for any input query $Q$ there exists a constant $c \in \mathbb{N}$ such that $P_o(n) \leq c \cdot P_u(n)$ holds. In other words, we will argue that the number of plans generated when using the seeding rule (i.e., when applying the proposed optimizations) is within a \textit{constant factor} of the number of plans generated without using the seeding rule. From this, it follows that the complexity of the optimization procedure when applying the proposed optimizations is within a constant factor of the complexity of the procedure without the proposed optimizations. This is because the sizes of any $\mathcal{T}_{Q,u}$ and $\mathcal{T}_{Q,o}$ are proportional to the sizes of their respective sets of leaves (i.e., numbers of concrete plans generated) multiplied by the size of $Q$ because the length of any path from the root of such a tree to a leaf is directly proportional to the size of $Q$.

The question remains \textit{how} to fix the shape to guarantee that if $P_o(n) \leq c \cdot P_u(n)$ holds for any query of this shape, that it indeed holds for any query in general. To do this, it must \textit{maximize} $c$ amongst all possible queries. In other words, the shape must be such that it increases $P_o(n)$ as much as possible \textit{relative to} $P_u(n)$.

The worst-case query shape, from the perspective of maximizing $c$, is a star-shaped query consisting \textit{only} of recursive terms. This query shape is shown in \autoref{fig:query_star_recursive}.
\begin{figure}[ht]
    \centering
     \begin{subfigure}{0.4\columnwidth}
         \includegraphics[width=\linewidth]{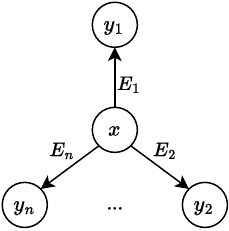}
         \caption{non-recursive}
          \label{fig:query_star}
     \end{subfigure}
     \hfill
     \begin{subfigure}{0.4\columnwidth}
         \includegraphics[width=\linewidth]{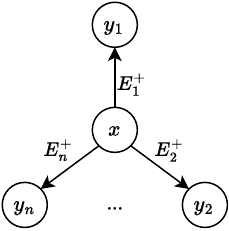}
         \caption{recursive}
         \label{fig:query_star_recursive}
     \end{subfigure}
    \caption{Star-shaped queries of $n$ terms (\autoref{sec:complexity})}
    \label{fig:query_stars}
\end{figure}
This shape maximizes $c$ because it maximizes the number of opportunities for the seeding rule to be applied. The reason for this is twofold. Firstly, because all terms in this shape are recursive, the seeding rule applies to any set of terms of size at least two. Secondly, the star-shape maximizes the rate at which the number of possible plans grows, and this rate of growth is proportional to the number of sub-queries of the input query \cite{ono90}, which is in turn proportional to the number of times the seeding rule can be applied.

We derive $P_u(n)$ and $P_o(n)$ based on the query shapes presented in \autoref{fig:query_star} and \autoref{fig:query_star_recursive}, respectively. By deriving $P_u(n)$ based on the non-recursive query in \autoref{fig:query_star}, we ensure by construction that the seeding rule does not apply.

\subsubsection{Analysis of non-recursive star-shaped queries}
The query shape from \autoref{fig:query_star} can be analyzed using two templates $T_1$ and $T_2$. Template $T_1$ represents abstractions of $k$ terms arranged in a star, with $2 \leq k \leq n$. That is, with $k = n$ this template represents the input query. Template $T_2$ represents abstractions over a single term $E_i$. The templates are specified as follows:
\begin{equation}
    \begin{aligned}
        & T_1 :\, E_1...E_k \quad \text{for } 2 \leq k \leq n &&& T_2 :\, E_1
    \end{aligned}
\end{equation}
Using these templates we can derive $P_u(n)$. We observe that there are $n$ ways to instantiate $T_2$ because there are $n$ terms to pick from in the input query. Each instantiation yields a single plan because a single term cannot be broken down further. Similarly, there are $\binom{n}{k}$ ways to instantiate $T_1$ because we can pick any $k$ out of $n$ terms. There are $2^{k-1} - 1$ non-empty, strict subsets of a set of $k$ elements when not counting a subset and its complement separately. Hence each instantiation of $T_1$ yields $2^{k-1} - 1$ plans.

Let $T_i(n, k)$ denote the number of plans generated by template $T_i$ for sizes $n$ and $k$. The formula $P_u(n)$ can be defined in terms of $T_1(n, k)$ and $T_2(n, k)$ as follows:
\begin{align}
    P_u(n) & = T_2(n, 1) + \sum\limits_{k=2}^n T_1(n, k) \nonumber \\
    & = n + \sum\limits_{k=2}^n \tbinom{n}{k}(2^{k-1} - 1) \nonumber \\
    & = \frac{1}{2}(3^n - 2^{n+1} +2n + 1)
\end{align}

\subsubsection{Analysis of recursive queries}
The query shape from \autoref{fig:query_star_recursive} can be analyzed using the previously defined templates $T_1$ and $T_2$ alongside two new templates. Templates $T_3$ and $T_4$ are similar to $T_1$ and $T_2$ respectively, but are used to represent recursive terms $E_i^+$ or sequences thereof. The templates are specified as follows:
\begin{equation}
    \begin{aligned}
        T_3 :&\, E_1^+...E_k^+ \quad \text{for } 2 \leq k \leq n &&
        T_4 :&\, E_1^+
    \end{aligned}
\end{equation}
There are $n$ ways to instantiate $T_4$ and each instantiation generates one plan which captures a fix-point procedure for $E_1^+$. Template $T_3$ can be instantiated in $\binom{n}{k}$ ways, just as $T_1$ was. The number of plans generated for each instantiation of $T_3$ is $2^{k-1}$, with $2^{k-1} - 1$ generated by the join rule and one plan by the seeding rule. The formula $P_o(n)$ can be defined as follows:
\begin{align}
    P_o(n) & = T_2(n, 1) + T_4(n,1) + \sum\limits_{k=2}^n T_1(n, k) + T_3(n, k) \nonumber \\
    & = 2n + \sum\limits_{k=2}^n \tbinom{n}{k}(2^{k-1} - 1) + \tbinom{n}{k}2^{k-1} \nonumber \\
    & = 2n + \sum\limits_{k=2}^n \tbinom{n}{k}(2^k - 1) \nonumber \\
    & = 3^n + 2^{n-1}(n-2) + 3n
\end{align}
Having defined $P_u(n)$ and $P_o(n)$ we can relate both formulas and conclude that the proposed optimizations can be applied with only a constant factor increase to the complexity of the optimization procedure.
\begin{theorem}
    For all $n \geq 2$ it holds that:
    \begin{align}
        P_o(n) &\leq 6 \cdot P_u(n) \quad \Leftrightarrow \nonumber  \\
        3^n + 2^{n-1}(n-2) + 3n & \leq 6 \cdot \Bigl(\frac{1}{2}(3^n - 2^{n+1} +2n + 1)\Bigr) \nonumber
    \end{align}
    \label{theo:enum}
\end{theorem}
\begin{proof}
    By a straight-forward induction on $n$.
\end{proof}
\begin{corollary}
    The optimization procedure's complexity when applying the proposed optimizations is within a constant factor of its complexity when not applying these optimizations.
\end{corollary}

\subsection{Cost-based optimization}
\label{sec:estimation}
While the enumerator generates a set of query plans for a given input query, only a single plan from this set is ultimately executed in order to obtain the query's result. Choosing which one plan to execute is commonly done on the basis of some \textit{estimated cost}. The estimated cost of a plan represents the estimated amount of time and possibly other resources such as memory and disk I/O required to execute the plan. A plan with lower estimated cost is to be preferred over an alternative plan with higher estimated cost at any stage during the optimization procedure. Producing high quality cost estimates (and \textit{cardinality} estimates as a component thereof) is a very challenging problem \cite{DBLP:journals/tkde/LeeuwenFY23,10.14778/3503585.3503586}, not only in the graph- but also in the relational paradigm, and it is something that even mature database systems often struggle with \cite{DBLP:journals/vldb/LeisRGMBKN18}. As such, an in-depth discussion of cardinality- and cost estimation for Regular Queries on graphs and the type of query plans we have proposed is beyond the scope of this work.

Nevertheless, our approach does employ cost-based optimization and therefore must provide some form of cost estimation. The estimation methods are similar to existing approaches \cite{DBLP:conf/sigmod/YakovetsGG16,DBLP:conf/sigmod/JachietGGL20} insofar as that they rely on a \textit{catalog} of stored facts and statistics about the database instance, such as the number of edges in the graph, the number of edges with a certain label for each label in the graph and synopses of the sets of nodes that have edges with a certain label incoming on- or outgoing from them. 
The information maintained in the catalog can then be combined to produce a \textit{cardinality} estimate for any arbitrary query. Given a logical operator in a query plan, the \textit{cost} of executing the operator can subsequently be estimated based on the operator's type and the cardinality estimate(s) for its input(s). The way in which our system uses the catalog information is similar to the approach taken by \postgres{}~\footnote{https://www.postgresql.org/docs/current/planner-stats-details.html}.

\section{Experimental Study}
\label{sec:method}

We next investigate the empirical behavior of our new methods, demonstrating that they enable orders of magnitude improvements in query evaluation performance.
Three main conclusions can be drawn from this experimental study, which are summarized below.

\paragraph{Conclusion (1): the proposed optimizations are potent.}
By which we mean that the performance improvements offered by the proposed optimizations are significant. Note that this concerns the optimizations themselves, and not yet their application in a cost-based optimizer. We conclude that the performance of the \quotes{best plan in practice} that uses any of the proposed optimizations, is often significantly better than the performance of the best plan in practice amongst those plans that do \textit{not} use any of the proposed optimizations. Whether the optimizations can be applied effectively in a cost-based optimizer is investigated separately. We elaborate on the notion of the best plan in practice in \autoref{sec:method_setup}.

\paragraph{Conclusion (2): the proposed optimizations can be applied effectively.} In order to apply an optimization technique effectively its application must not have a detrimental effect on the performance of the optimization procedure itself. Moreover, the improvement to query evaluation performance achieved through cost-based optimization (i.e., choosing an estimated best plan, rather than necessarily obtaining the best plan in practice) must be significant.

\begin{table*}[ht]
\centering
\resizebox{\textwidth}{!}{
\begin{tabular}{ll||llll||llll||llll||llll||}
\multicolumn{2}{l||}{Template} & \multicolumn{4}{l||}{CCC (\#instances = 24)} & \multicolumn{4}{l||}{PCC2 (\#instances = 16)} & \multicolumn{4}{l||}{PCC3 (\#instances = 4)} & \multicolumn{4}{l||}{All (\#instances = 44)} \\
\multicolumn{2}{l||}{Metric} & PC & AC & PT & AT & PC & AC & PT & AT & PC & AC & PT & AT & PC & AC & PT & AT \\ \hline \hline
\multicolumn{2}{l||}{Min} & 0.536 & \cellcolor[HTML]{DAE8FC}\dotuline{0.419 } & 1 & \cellcolor[HTML]{DAE8FC}0.625 & 0.405 & \cellcolor[HTML]{DAE8FC}\dotuline{0.105 } & 1 & \cellcolor[HTML]{DAE8FC}0.227 & 1.91 & \cellcolor[HTML]{DAE8FC}\dotuline{0.552 } & 1.9 & \cellcolor[HTML]{DAE8FC}0.761 & 0.405 & \cellcolor[HTML]{DAE8FC}0.105 & 1 & \cellcolor[HTML]{DAE8FC}0.227 \\ \hline
 & 10 & 1 & \cellcolor[HTML]{DAE8FC}1 & 1.02 & \cellcolor[HTML]{DAE8FC}0.679 & 1 & \cellcolor[HTML]{DAE8FC}0.405 & 1 & \cellcolor[HTML]{DAE8FC}0.581 & 1.91 & \cellcolor[HTML]{DAE8FC}0.552 & 1.9 & \cellcolor[HTML]{DAE8FC}0.761 & 1 & \cellcolor[HTML]{DAE8FC}0.552 & 1 & \cellcolor[HTML]{DAE8FC}0.679 \\
 & 25 & 1.13 & \cellcolor[HTML]{DAE8FC}1.07 & 1.3 & \cellcolor[HTML]{DAE8FC}0.763 & 1 & \cellcolor[HTML]{DAE8FC}1 & 1 & \cellcolor[HTML]{DAE8FC}0.885 & 1.91 & \cellcolor[HTML]{DAE8FC}0.552 & 1.9 & \cellcolor[HTML]{DAE8FC}0.761 & 1 & \cellcolor[HTML]{DAE8FC}1 & 1.13 & \cellcolor[HTML]{DAE8FC}0.862 \\
 & 50 & 1.77 & \cellcolor[HTML]{DAE8FC}1.68 & \dashuline{1.98 } & \cellcolor[HTML]{DAE8FC}\dashuline{1.89 } & 1 & \cellcolor[HTML]{DAE8FC}1 & \dashuline{1.27 } & \cellcolor[HTML]{DAE8FC}\dashuline{1.11 } & 4.25 & \cellcolor[HTML]{DAE8FC}2.34 & \dashuline{4.87 } & \cellcolor[HTML]{DAE8FC}\dashuline{3.09 } & 1.77 & \cellcolor[HTML]{DAE8FC}1.59 & \dashuline{1.9 } & \cellcolor[HTML]{DAE8FC}\dashuline{1.88 } \\
 & 75 & 5.52 & \cellcolor[HTML]{DAE8FC}3.21 & 6.43 & \cellcolor[HTML]{DAE8FC}3.48 & 2.1 & \cellcolor[HTML]{DAE8FC}2.1 & 3.52 & \cellcolor[HTML]{DAE8FC}3.41 & 12.9 & \cellcolor[HTML]{DAE8FC}2.94 & 12.9 & \cellcolor[HTML]{DAE8FC}3.57 & 4.32 & \cellcolor[HTML]{DAE8FC}3 & 4.87 & \cellcolor[HTML]{DAE8FC}3.57 \\
\multirow{-5}{*}{\rotatebox{90}{Percentile}} & 90 & 11.7 & \cellcolor[HTML]{DAE8FC}6.23 & 13 & \cellcolor[HTML]{DAE8FC}6.93 & 5.23 & \cellcolor[HTML]{DAE8FC}5.24 & 6.4 & \cellcolor[HTML]{DAE8FC}5.87 & 70.5 & \cellcolor[HTML]{DAE8FC}2.94 & 80 & \cellcolor[HTML]{DAE8FC}3.68 & 11.7 & \cellcolor[HTML]{DAE8FC}5.24 & 12.9 & \cellcolor[HTML]{DAE8FC}6.9 \\ \hline
\multicolumn{2}{l||}{Max} & 74.1 & \cellcolor[HTML]{DAE8FC}75 & 75 & \cellcolor[HTML]{DAE8FC}\uline{59.1 } & 11.7 & \cellcolor[HTML]{DAE8FC}11.1 & 11.6 & \cellcolor[HTML]{DAE8FC}7.23 & 70.5 & \cellcolor[HTML]{DAE8FC}2.94 & 80 & \cellcolor[HTML]{DAE8FC}3.68 & 74.1 & \cellcolor[HTML]{DAE8FC}75 & 80 & \cellcolor[HTML]{DAE8FC}59.1 \\
\multicolumn{2}{l||}{Mean} & 6.42 & \cellcolor[HTML]{DAE8FC}5.38 & 6.93 & \cellcolor[HTML]{DAE8FC}5.35 & 2.38 & \cellcolor[HTML]{DAE8FC}2.25 & 2.74 & \cellcolor[HTML]{DAE8FC}2.28 & 22.4 & \cellcolor[HTML]{DAE8FC}2.19 & 24.9 & \cellcolor[HTML]{DAE8FC}2.78 & 6.4 & \cellcolor[HTML]{DAE8FC}3.95 & 7.04 & \cellcolor[HTML]{DAE8FC}4 \\
\multicolumn{18}{l}{\small PC (PT) = potential improvement in terms of cardinality (query time), AC (AT) = actual improvement in terms of cardinality (query time)}
\end{tabular}
}
\caption{Potential- and minimal actual improvements on \stringdb{} (rounded to 3 significant digits).
Median improvements in terms of query evaluation time (\pottime, \acttime) are \dashuline{significant}.
Maximal improvements in terms of time (\acttime) are \uline{up to an order of magnitude}.
For some instances, \dotuline{more data is processed} by optimized plans without large detriment to evaluation time. (\autoref{sec:experiments-potency},\autoref{sec:experiments-effective})}
\label{tab:improvements_stringdb}
\end{table*}

\paragraph{Conclusion (3): application of the proposed optimizations enables significant speed-ups of up to several orders of magnitude in terms of query evaluation performance over state-of-the-art systems.}
We place the performance figures obtained using the proposed optimizations in the context of the performance achieved by state-of-the-art systems.

\subsection{Methodology and setup}
\label{sec:method_setup}
Some terminology is required to describe the methodology and setup. For an arbitrary query $Q$ and plan $p$ we define:
\begin{itemize}
    \setlength\itemsep{0.5em}
    \item $\uq$ the set of \textit{unoptimized} plans for $Q$ (i.e., plans that do not use any of the proposed optimizations),
    \item $\sq$ the set of \textit{optimized} plans for $Q$ (i.e., plans that use at least one of the proposed optimizations),
    \item $\ebpu \in \uq$ the \textit{estimated} best unoptimized plan,
    \item $\ebpo \in \sq$ the \textit{estimated} best optimized plan,
    \item $\abpu \in \uq$ the best unoptimized plan \textit{in practice},
    \item $\abpo \in \sq$ the best optimized plan \textit{in practice},
    \item $c(p)$ the \textit{total} number of tuples (cardinality) processed during execution of $p$,
    \item $\timeplan{p}$ the \textit{optimization} time (ms) to obtain $p$,
    \item $\timeexec{p}$ the \textit{execution} time (ms) of $p$,
    \item $t(p) = \timeplan{p} + \timeexec{p}$ the \textit{evaluation} time of $p$,
    \item $\dfrac{c(\abpu)}{c(\abpo)}$ the \textit{potential} improvement (\potcard{}) in terms of cardinality (i.e., total number of tuples processed),
    \item $\dfrac{t(\abpu)}{t(\abpo)}$ the \textit{potential} improvement (\pottime{}) in terms of query evaluation time,
    \item $\dfrac{c(\abpu)}{c(\ebpo)}$ the \textit{minimal actual} improvement (\actcard{}) in terms of cardinality, and
    \item $\dfrac{t(\abpu)}{t(\ebpo)}$ the \textit{minimal actual} improvement (\acttime{}) in terms of query evaluation time.
\end{itemize}
The best plan \textit{in practice} is the plan that exhibited the lowest wall-clock evaluation time out of a set of plans. This observation is made by exhaustive executing all plans within a set. The \textit{estimated} best plan is the plan with the lowest estimated cost out of a set of plans, based on some cost model.

By \dquotes{total number of tuples processed} we mean the sum of all numbers of tuples output by operators that generate \textit{new} tuples (i.e., reading edge-triples, performing joins, etc.) while ignoring the tuples output by operators that only \quotes{forward} tuples (i.e., unions, projections, selections, renames, et cetera).

Note that the metrics \actcard{} and \acttime{} are \textit{minimal} in the sense that the improvement of $\ebpo$ over any $p \in \uq$ is at least as large as the improvement over $\abpu$. This makes the baseline a hypothetical system that always picks the best possible plan from $\uq$. Therefore, the setup is \textit{disadvantageous} to our proposed system on this point making the results obtained \textit{conservative}. 

The reason for looking at total cardinality alongside query evaluation time, is that the former offers an \textit{implementation-independent} view on the savings made through the use of the proposed optimizations. That is, if less data is being processed in order to evaluate a query when the optimizations are used (i.e., \potcard{} and \actcard{} are larger than one), this is beneficial regardless of the system on which the evaluation is performed.

We investigate the \potcard{} and \pottime{} metrics for a multitude of queries on two diverse datasets. To show that the proposed optimizations can be applied effectively we look at the time $\timeplan{\ebpo}$ spent to obtain $\ebpo$ for various queries, and we consider the metrics \actcard{} and \acttime{}. Finally, to place the performance of our system with- and without the use of the proposed optimizations in context w.r.t. state-of-the-art systems we investigate the query evaluation time $t(\ebpo)$ for various queries and datasets. Note that $\ebpo$ is defined w.r.t. each system.

The remainder of this section is structured as follows. The workloads, in terms of query templates, instances and datasets are discussed in \autoref{sec:method-workloads}. The experimental environment and the systems under consideration are the subject of \autoref{sec:method-environment} and \autoref{sec:method-systems}. Finally, the experiments supporting each of the three main conclusions are discussed in \autoref{sec:experiments}.

\subsection{Workloads}
\label{sec:method-workloads}
In our experimental study, we adopt the methodology used by state-of-the-art studies~\cite{DBLP:conf/sigmod/YakovetsGG16,DBLP:conf/sigmod/JachietGGL20,arroyuelo2023optimizing}. We define benchmarking workloads by \emph{mining} query instances based on predefined templates from a variety of datasets across diverse application domains. The mined queries were then used to evaluate the performance of the query evaluation engines of the systems under study.

The distinction between query templates and instances are outlined in \autoref{sec:templates_and_instances}, and the query templates used in the experimental evaluation are presented. The datasets from which query instances are mined are the topic of \autoref{sec:datasets}. The experimental environment and competitor systems are discussed in \autoref{sec:method-environment} and \autoref{sec:method-systems}.

\subsubsection{Query templates and instances}
\label{sec:templates_and_instances}
We will use the terms \textit{query template} and \textit{query instance} (or simply template and instance) throughout the experimental section (not to be confused with the templates used for analysis in \autoref{sec:complexity}). A query template is a generic specification of a graph query that lacks constants such as edge labels and filter values. We mine combinations of edge labels and filter values from real property graphs to construct concrete query instances. From the set of query instances mined in this way, we identify a subset of \textit{valid} instances. An instance if valid if and only if it satisfies the following three criteria:
\begin{enumerate}
    \item the instance must have a \textit{non-empty} result on the dataset from which it was mined,
    \item evaluation of the instance must terminate successfully on at least one system, and
    \item evaluation of the instance using $\ebpu$ on our system must take at least one second.
\end{enumerate}
The third criterion excludes queries that are already \quotes{fast} and focuses instead on those queries worth optimizing. Moreover, any performance gains that could be made on queries which are already \quotes{sub-second} are more difficult to attribute wholly to the proposed optimizations. All query instances are converted to \textit{count} queries. That is, instead of materializing the full set of results, we only count the number of results. We do \textit{not} employ optimizations for count queries in particular anywhere in our system. Definitions for the templates are given in \autoref{eq:template_and_base_definitions}.

These templates were chosen because they are \textit{fundamental}, \textit{simple} and \textit{intuitive}. 
They are fundamental because they capture the essence of the navigational queries that are commonly used in practice~\cite{wdbench}.
They are simple because they are composed of few relations and filters. 
They are intuitive because they capture intuitive natural language questions which combine \quotes{connectedness} (i.e., reachability) and conjunction.
For example, the CCC and PCC templates ask for \dquotes{connected components which have something in common} and \dquotes{components which are connected in more than one way}, respectively. 
The variables $l_1, l_2, l_3$ and $c_1$ are placeholders for edge labels and a filter value. An exhaustive list of all query instances used and their specification in SPARQL and SQL can be found here \cite{Ymous_Optimizing_Navigational_Queries}.

\begin{equation}
\begin{aligned}[b]
    & R(s, t) \leftarrow E(s, e, t), P(e, \text{label}, l_1) \\
    & S(s, t) \leftarrow E(s, e, t), P(e, \text{label}, l_2)  \\
    & T(s, t) \leftarrow E(s, e, t), P(e, \text{label}, l_3) \\
    \hfill \\
    & CCC1(x, y, z) \leftarrow R^+(x, y), S(x, z), T(z, y) \\
    & CCC2(x, y, z) \leftarrow R^+(x, y), S(x, z), T(y, z) \\
    & CCC3(x, y, z) \leftarrow R^+(x, y), S(z, x), T(z, y) \\
    & CCC4(x, y, z) \leftarrow R^+(x, y), S(z, x), T(y, z) \\
    \hfill \\
    & PCC2(x, y) \leftarrow R^+(x, y), S^+(x, y) \\
    & PCC3(x, y) \leftarrow R^+(x, y), S^+(x, y), T^+(x, y) \\
    \hfill \\
    &I(x, y) \leftarrow S(x, y), T^+(x, z), z = c_1 \\
    &RQ(x, y, z) \leftarrow R(x, y), I^+(y, z) \\
\end{aligned}
\label{eq:template_and_base_definitions}
\end{equation}

\subsubsection{Datasets}
\label{sec:datasets}
We conduct experiments on two real-world datasets coming from diverse application domains, namely DBPedia's 2020 release (\dbpedia)\cite{DBLP:journals/semweb/LehmannIJJKMHMK15,DBLP:conf/i-semantics/HoferHDF20} and the sub-graph related to humans of the Search Tool for the Retrieval of Interacting Genes/Proteins (\stringdb)\cite{stringdb/gkac1000}. \dbpedia{} is a \textit{knowledge graph} composed of information created in various Wikimedia projects. \stringdb{} is a database of known and predicted protein-protein interactions, which can be searched on a per-organism basis. The full 2020 release of \dbpedia{} contains upwards of $483$M edges and features almost $55$K distinct edge labels. The sub-graph of \stringdb{} pertaining to humans contains just over 1.5M edges (protein-protein interactions)\footnote{\label{footnote:ccc}Protein-protein interactions are symmetric. As a consequence the query templates CCC1, CCC2, CCC3 and CCC4 all default to the same semantic meaning and so are referred to as CCC for \stringdb{}.} and 19.5K nodes (distinct proteins). Humans are the most frequently queried organism in practice on \stringdb's online access point \cite{Web:sting-statistics}. The sub-graph of human proteins is particularly \textit{dense}, which is challenging when leveraging the selectivity of join-predicates.

In order to represent RDF data such as \dbpedia{} as a property graph, we perform the following transformation. Let $D$ be set of RDF-triples of the form $(s, p, o)$ (i.e., subject, predicate, object), and define $S = \bigcup_{(s, p, o) \in D} \{ s \}$ and \added[id=R1,comment={D5}]{$T = \bigcup_{(s, p, o) \in D} \{ o \}$}. Let $\eta: S \cup T \rightarrow \mathbb{N}$ be a function that assigns every subject and object in $D$ a distinct numerical identifier, and let $\zeta: D \rightarrow \mathbb{N}$ be a function that assigns every triple in $D$ a distinct numerical identifier, such that the ranges of $\eta$ and $\zeta$ are disjoint. Then we can construct the sets $E$ and $P$ that comprise a property graph $G = (E, P)$ from $D$ by adding $(\eta(s), \zeta(s, p, o), \eta(o))$ to $E$ and adding $(\zeta(s, p, o), \text{label}, p)$ to $P$ for every $(s, p, o) \in D$.

\subsubsection{Environment}
\label{sec:method-environment}
All systems are deployed in the same environment and evaluation times are averaged over five runs. Plans that take more than a set time limit to complete are timed-out. This time limit is 2 minutes for \stringdb{} and 10 minutes for \dbpedia{}. A lower limit is used for \stringdb{} because the average query evaluation time is higher on \stringdb{} due to its density and as a consequence exploring plan spaces exhaustively with a greater time limit is infeasible in practice\footnote{As a further consequence, exhaustively exploring $\sq$ to find $\abpo$ is infeasible for the RQ template on \stringdb{}.}.
The experiments are performed on a server running Ubuntu 20.04 with two Intel Xeon E5-2697 v2 CPUs (12 cores each, 2 threads per core at 2.7GHz) and 492 GB of RAM. The amount of RAM available to each database system is limited to 320 GB.

\subsubsection{Systems}
\label{sec:method-systems}
We evaluate the workloads on various state-of-the-art systems, namely: \avantgraphu{} (our system with optimizations disabled (\agu)),~\avantgraphs{} (our system with optimizations enabled (\ags)), \millenniumdbv{} (\mdb),~\duckdbv{} (\ddb),~\postgresv{} (\pg), and~\virtuosov{} (\vt). Moreover, we include a version of our system (\avantgraphwg{}, (\agwg{})) which uses only state-of-the-art seeding techniques as captured by Waveguide~\cite{DBLP:conf/sigmod/YakovetsGG16} (i.e., it does \textit{not} use seeding of interior closures nor stacking). We employ these techniques in our system, rather than invoking Waveguide itself for two reasons. First, the available implementation of Waveguide does not support conjunctive queries (but its optimization techniques still apply to conjunctive queries) and second, its evaluation procedure relies on a procedural SQL routine which is a bottleneck on performance. Hence we believe that reproducing Waveguide's optimization techniques on our system presents them in the best possible light.

Taken together these systems provide a representative view on the performance of the query workloads on state-of-the-art systems. Furthermore, we note that \kuzu{} \cite{kuzu:cidr,kuzu-github} does not allow for the expression of transitive closures, we could not obtain binaries for \umbra{} \cite{DBLP:conf/cidr/NeumannF20/umbra} and \tigergraph{} \cite{DBLP:journals/corr/abs-1901-08248/tigergraph} upon request, experienced blocking technical difficulties in running \duckpgq{} \cite{DBLP:conf/cidr/WoldeSSB23/duckpgq,Raasveldt_DuckDB/duckpgq}.

Systems are configured as per vendor recommendations. For the relational database systems \duckdb{} and \postgres{} property graphs are stored as a table of serialized triples (subject (\texttt{S}), predicate (\texttt{P}), object (\texttt{O})) with several indexes (\texttt{SPO}, \texttt{PSO}, \texttt{OSP}).
For the SPARQL-based systems \millenniumdb{} and \virtuoso{} property graphs are stored- and indexed using each system's defaults for graph-structured data. Finally, \avantgraph{} stores property graphs in a compressed sparse row (CSR) format and two types of simple indexes are constructed, namely one index facilitating the retrieval of all edges associated with a key edge label and one index facilitating the retrieval of all edges adjacent to- and associated with a key vertex and edge label.

\subsection{Experiments}
\label{sec:experiments}

\subsubsection{The proposed optimizations are potent}
\label{sec:experiments-potency}

The potential improvements in terms of cardinality (\potcard{}) and 
evaluation time (\pottime{}) are tabulated for \stringdb{} and \dbpedia{} in \autoref{tab:improvements_stringdb} and \autoref{tab:improvements_dbpedia}, respectively. We observe that the potential is high on \stringdb{} for the PCC3 template (median values of 4.25x and 4.87x), and significant for the other templates (up to 1.98x for \pottime{}).
For \dbpedia{}, we observe that the potential is high (median values in the range from 2.77x to 6.57x) for both \potcard{} and \pottime{} on the CCC and RQ templates. Moreover, the potential for improvement is extreme (median values in the range from 118x to 702x) for both \potcard{} and \pottime{} on the PCC2 and PCC3 templates. 

There several instances of the PCC2 and PCC3 templates on both datasets for which no $\abpu$ can be identified (i.e., all plans $p \in \uq$ are timed-out). As a consequence, the necessary metrics cannot be computed for these instances. It is important to note that the proposed optimizations enable the evaluation of many such instances within the time limit, as is shown in \autoref{tab:timeout}.

\begin{table*}[ht]
\centering
\resizebox{\textwidth}{!}{
\begin{tabular}{lrllllllllllllllll}
\multicolumn{2}{r||}{Template} & \multicolumn{4}{l||}{CCC1 (\# instances = 7)} & \multicolumn{4}{l||}{CCC2 (\# instances = 11)} & \multicolumn{4}{l||}{CCC3 (\# instances = 8)} & \multicolumn{4}{l||}{CCC4 (\# instances = 7)} \\
\multicolumn{2}{r||}{Metric} & \potcard & \actcard & \pottime & \multicolumn{1}{l||}{\acttime} & \potcard & \actcard & \pottime & \multicolumn{1}{l||}{\acttime} & \potcard & \actcard & \pottime & \multicolumn{1}{l||}{\acttime} & \potcard & \actcard & \pottime & \multicolumn{1}{l||}{\acttime} \\ \hline \hline
\multicolumn{2}{r||}{Min} & 2.83 & \cellcolor[HTML]{FFF2E3}0.871 & 2.4 & \multicolumn{1}{l||}{\cellcolor[HTML]{FFF2E3}2.31} & 1.05 & \cellcolor[HTML]{FFF2E3}0.582 & 1.41 & \multicolumn{1}{l||}{\cellcolor[HTML]{FFF2E3}0.662} & 1.16 & \cellcolor[HTML]{FFF2E3}0.879 & 1.82 & \multicolumn{1}{l||}{\cellcolor[HTML]{FFF2E3}0.749} & 1.17 & \cellcolor[HTML]{FFF2E3}0.541 & 1.55 & \multicolumn{1}{l||}{\cellcolor[HTML]{FFF2E3}0.741} \\ \hline
 & \multicolumn{1}{r||}{10} & 2.83 & \cellcolor[HTML]{FFF2E3}0.871 & 2.4 & \multicolumn{1}{l||}{\cellcolor[HTML]{FFF2E3}2.31} & 1.8 & \cellcolor[HTML]{FFF2E3}0.808 & 2.5 & \multicolumn{1}{l||}{\cellcolor[HTML]{FFF2E3}1.94} & 1.16 & \cellcolor[HTML]{FFF2E3}0.879 & 1.82 & \multicolumn{1}{l||}{\cellcolor[HTML]{FFF2E3}0.749} & 1.17 & \cellcolor[HTML]{FFF2E3}0.541 & 1.55 & \multicolumn{1}{l||}{\cellcolor[HTML]{FFF2E3}0.741} \\
 & \multicolumn{1}{r||}{25} & 3.43 & \cellcolor[HTML]{FFF2E3}1.23 & 2.57 & \multicolumn{1}{l||}{\cellcolor[HTML]{FFF2E3}2.42} & 2.52 & \cellcolor[HTML]{FFF2E3}0.827 & 2.54 & \multicolumn{1}{l||}{\cellcolor[HTML]{FFF2E3}2.14} & 2.99 & \cellcolor[HTML]{FFF2E3}0.991 & 2.33 & \multicolumn{1}{l||}{\cellcolor[HTML]{FFF2E3}2.3} & 3.22 & \cellcolor[HTML]{FFF2E3}0.828 & 2.53 & \multicolumn{1}{l||}{\cellcolor[HTML]{FFF2E3}2.03} \\
 & \multicolumn{1}{r||}{50} & 6.2 & \cellcolor[HTML]{FFF2E3}1.82 & \dashuline{4.88 } & \multicolumn{1}{l||}{\cellcolor[HTML]{FFF2E3}\dashuline{2.55 }} & 4.02 & \cellcolor[HTML]{FFF2E3}1.3 & \dashuline{4.85 } & \multicolumn{1}{l||}{\cellcolor[HTML]{FFF2E3}\dashuline{4.31 }} & 3.18 & \cellcolor[HTML]{FFF2E3}1.1 & \dashuline{2.77 } & \multicolumn{1}{l||}{\cellcolor[HTML]{FFF2E3}\dashuline{2.41 }} & 3.56 & \cellcolor[HTML]{FFF2E3}1.02 & \dashuline{4.55 } & \multicolumn{1}{l||}{\cellcolor[HTML]{FFF2E3}\dashuline{2.34 }} \\
 & \multicolumn{1}{r||}{75} & 7.94 & \cellcolor[HTML]{FFF2E3}2.81 & 5.11 & \multicolumn{1}{l||}{\cellcolor[HTML]{FFF2E3}2.79} & 22.5 & \cellcolor[HTML]{FFF2E3}9.35 & 147 & \multicolumn{1}{l||}{\cellcolor[HTML]{FFF2E3}21.6} & 3.79 & \cellcolor[HTML]{FFF2E3}1.15 & 4.68 & \multicolumn{1}{l||}{\cellcolor[HTML]{FFF2E3}2.54} & 6.77 & \cellcolor[HTML]{FFF2E3}1.85 & 4.77 & \multicolumn{1}{l||}{\cellcolor[HTML]{FFF2E3}2.42} \\
\multirow{-5}{*}{\rotatebox{90}{Percentile}} & \multicolumn{1}{r||}{90} & 8.1 & \cellcolor[HTML]{FFF2E3}3.24 & 6.28 & \multicolumn{1}{l||}{\cellcolor[HTML]{FFF2E3}4.02} & 87.8 & \cellcolor[HTML]{FFF2E3}24.9 & 160 & \multicolumn{1}{l||}{\cellcolor[HTML]{FFF2E3}34.7} & 19.7 & \cellcolor[HTML]{FFF2E3}8.87 & 100 & \multicolumn{1}{l||}{\cellcolor[HTML]{FFF2E3}21.8} & 10.8 & \cellcolor[HTML]{FFF2E3}1.96 & 5.42 & \multicolumn{1}{l||}{\cellcolor[HTML]{FFF2E3}2.44} \\ \hline
\multicolumn{2}{r||}{Max} & 8.1 & \cellcolor[HTML]{FFF2E3}3.24 & 6.28 & \multicolumn{1}{l||}{\cellcolor[HTML]{FFF2E3}4.02} & 109 & \cellcolor[HTML]{FFF2E3}56.3 & 181 & \multicolumn{1}{l||}{\cellcolor[HTML]{FFF2E3}57.7} & 19.7 & \cellcolor[HTML]{FFF2E3}8.87 & 100 & \multicolumn{1}{l||}{\cellcolor[HTML]{FFF2E3}21.8} & 10.8 & \cellcolor[HTML]{FFF2E3}1.96 & 5.42 & \multicolumn{1}{l||}{\cellcolor[HTML]{FFF2E3}2.44} \\
\multicolumn{2}{r||}{Mean} & 5.73 & \cellcolor[HTML]{FFF2E3}1.95 & 4.15 & \multicolumn{1}{l||}{\cellcolor[HTML]{FFF2E3}2.73} & 22.5 & \cellcolor[HTML]{FFF2E3}9.59 & 47 & \multicolumn{1}{l||}{\cellcolor[HTML]{FFF2E3}12.5} & 5.23 & \cellcolor[HTML]{FFF2E3}2.14 & 15.5 & \multicolumn{1}{l||}{\cellcolor[HTML]{FFF2E3}4.67} & 4.8 & \cellcolor[HTML]{FFF2E3}1.19 & 3.97 & \multicolumn{1}{l||}{\cellcolor[HTML]{FFF2E3}2.1} \\
 & \multicolumn{1}{l}{} &  &  &  &  &  &  &  &  &  &  &  &  &  &  &  &  \\
\multicolumn{2}{r||}{Template} & \multicolumn{4}{l||}{PCC2 (\#instances = 57)} & \multicolumn{4}{l||}{PCC3 (\#instances = 4)} & \multicolumn{4}{l||}{RQ (\#instances = 127)} & \multicolumn{4}{l||}{All (\#instances = 221)} \\
\multicolumn{2}{r||}{Metric} & \potcard & \actcard & \pottime & \multicolumn{1}{l||}{\acttime} & \potcard & \actcard & \pottime & \multicolumn{1}{l||}{\acttime} & \potcard & \actcard & \pottime & \multicolumn{1}{l||}{\acttime} & \potcard & \actcard & \pottime & \multicolumn{1}{l||}{\acttime} \\ \hline \hline
\multicolumn{2}{r||}{Min} & 1.36 & \cellcolor[HTML]{FFF2E3}0.575 & 1.53 & \multicolumn{1}{l||}{\cellcolor[HTML]{FFF2E3}1.03} & 25 & \cellcolor[HTML]{FFF2E3}0.944 & 50.3 & \multicolumn{1}{l||}{\cellcolor[HTML]{FFF2E3}0.675} & 3.06 & \cellcolor[HTML]{FFF2E3}1.02 & 2.31 & \multicolumn{1}{l||}{\cellcolor[HTML]{FFF2E3}0.929} & 1.05 & \cellcolor[HTML]{FFF2E3}0.541 & 1.41 & \multicolumn{1}{l||}{\cellcolor[HTML]{FFF2E3}0.662} \\ \hline
 & \multicolumn{1}{r||}{10} & 17.4 & \cellcolor[HTML]{FFF2E3}2.23 & 19.2 & \multicolumn{1}{l||}{\cellcolor[HTML]{FFF2E3}2.26} & 25 & \cellcolor[HTML]{FFF2E3}0.944 & 50.3 & \multicolumn{1}{l||}{\cellcolor[HTML]{FFF2E3}0.675} & 5.03 & \cellcolor[HTML]{FFF2E3}1.97 & 2.75 & \multicolumn{1}{l||}{\cellcolor[HTML]{FFF2E3}1.09} & 4.02 & \cellcolor[HTML]{FFF2E3}1.1 & 2.74 & \multicolumn{1}{l||}{\cellcolor[HTML]{FFF2E3}1.12} \\
 & \multicolumn{1}{r||}{25} & 65.1 & \cellcolor[HTML]{FFF2E3}2.27 & 75.6 & \multicolumn{1}{l||}{\cellcolor[HTML]{FFF2E3}2.4} & 25 & \cellcolor[HTML]{FFF2E3}0.944 & 50.3 & \multicolumn{1}{l||}{\cellcolor[HTML]{FFF2E3}0.675} & 5.99 & \cellcolor[HTML]{FFF2E3}2.07 & 4.58 & \multicolumn{1}{l||}{\cellcolor[HTML]{FFF2E3}1.16} & 5.99 & \cellcolor[HTML]{FFF2E3}2.06 & 4.62 & \multicolumn{1}{l||}{\cellcolor[HTML]{FFF2E3}1.18} \\
 & \multicolumn{1}{r||}{50} & 141 & \cellcolor[HTML]{FFF2E3}43 & \uline{185 } & \multicolumn{1}{l||}{\cellcolor[HTML]{FFF2E3}\uline{42.9 }} & 118 & \cellcolor[HTML]{FFF2E3}1.89 & \uline{702 } & \multicolumn{1}{l||}{\cellcolor[HTML]{FFF2E3}\dashuline{1.42 }} & 6.57 & \cellcolor[HTML]{FFF2E3}2.11 & \dashuline{4.9 } & \multicolumn{1}{l||}{\cellcolor[HTML]{FFF2E3}\dashuline{1.2 }} & 22.5 & \cellcolor[HTML]{FFF2E3}2.17 & 28.7 & \multicolumn{1}{l||}{\cellcolor[HTML]{FFF2E3}1.32} \\
 & \multicolumn{1}{r||}{75} & 601 & \cellcolor[HTML]{FFF2E3}121 & 623 & \multicolumn{1}{l||}{\cellcolor[HTML]{FFF2E3}132} & 383 & \cellcolor[HTML]{FFF2E3}4.43 & 727 & \multicolumn{1}{l||}{\cellcolor[HTML]{FFF2E3}1.48} & 123 & \cellcolor[HTML]{FFF2E3}2.25 & 95.9 & \multicolumn{1}{l||}{\cellcolor[HTML]{FFF2E3}1.25} & 131 & \cellcolor[HTML]{FFF2E3}24.9 & 144 & \multicolumn{1}{l||}{\cellcolor[HTML]{FFF2E3}34.7} \\
\multirow{-5}{*}{\rotatebox{90}{Percentile}} & \multicolumn{1}{r||}{90} & 2000 & \cellcolor[HTML]{FFF2E3}766 & 1060 & \multicolumn{1}{l||}{\cellcolor[HTML]{FFF2E3}846} & 800 & \cellcolor[HTML]{FFF2E3}4.43 & 882 & \multicolumn{1}{l||}{\cellcolor[HTML]{FFF2E3}2.51} & 144 & \cellcolor[HTML]{FFF2E3}33.6 & 152 & \multicolumn{1}{l||}{\cellcolor[HTML]{FFF2E3}71.5} & 239 & \cellcolor[HTML]{FFF2E3}52.7 & 385 & \multicolumn{1}{l||}{\cellcolor[HTML]{FFF2E3}83.6} \\ \hline
\multicolumn{2}{r||}{Max} & 11300 & \cellcolor[HTML]{FFF2E3}2120 & 6370 & \multicolumn{1}{l||}{\cellcolor[HTML]{FFF2E3}\dotuline{1480 }} & 800 & \cellcolor[HTML]{FFF2E3}4.43 & 882 & \multicolumn{1}{l||}{\cellcolor[HTML]{FFF2E3}2.51} & 217 & \cellcolor[HTML]{FFF2E3}56.8 & 173 & \multicolumn{1}{l||}{\cellcolor[HTML]{FFF2E3}98.9} & 11300 & \cellcolor[HTML]{FFF2E3}2120 & 6370 & \multicolumn{1}{l||}{\cellcolor[HTML]{FFF2E3}1480} \\
\multicolumn{2}{r||}{Mean} & 748 & \cellcolor[HTML]{FFF2E3}230 & 563 & \multicolumn{1}{l||}{\cellcolor[HTML]{FFF2E3}188} & 332 & \cellcolor[HTML]{FFF2E3}2.92 & 590 & \multicolumn{1}{l||}{\cellcolor[HTML]{FFF2E3}1.52} & 57.8 & \cellcolor[HTML]{FFF2E3}8.02 & 51.1 & \multicolumn{1}{l||}{\cellcolor[HTML]{FFF2E3}13} & 234 & \cellcolor[HTML]{FFF2E3}64.7 & 188 & \multicolumn{1}{l||}{\cellcolor[HTML]{FFF2E3}57} \\
\multicolumn{18}{l}{\small PC (PT) = potential improvement in terms of cardinality (query time), AC (AT) = actual improvement in terms of cardinality (query time)}
\end{tabular}
}
\caption{Potential- and minimal actual improvements on DBPedia (rounded to 3 significant digits).
Median improvements in terms of query evaluation time (\pottime, \acttime) range from \dashuline{significant} to \uline{extreme}. The largest actual improvements (\acttime) are \dotuline{up to three order of magnitude.} (\autoref{sec:experiments-potency},\autoref{sec:experiments-effective})}
\label{tab:improvements_dbpedia}
\end{table*}

\begin{figure*}[ht]
    \centering
     \begin{subfigure}[t]{0.32\textwidth}
         \includegraphics[width=\linewidth]{img_context_stringdb_ccc_pcc2_pcc3.pdf}
         \label{fig:context_stringdb}
     \end{subfigure}
     \begin{subfigure}[t]{0.32\textwidth}
         \includegraphics[width=\linewidth]{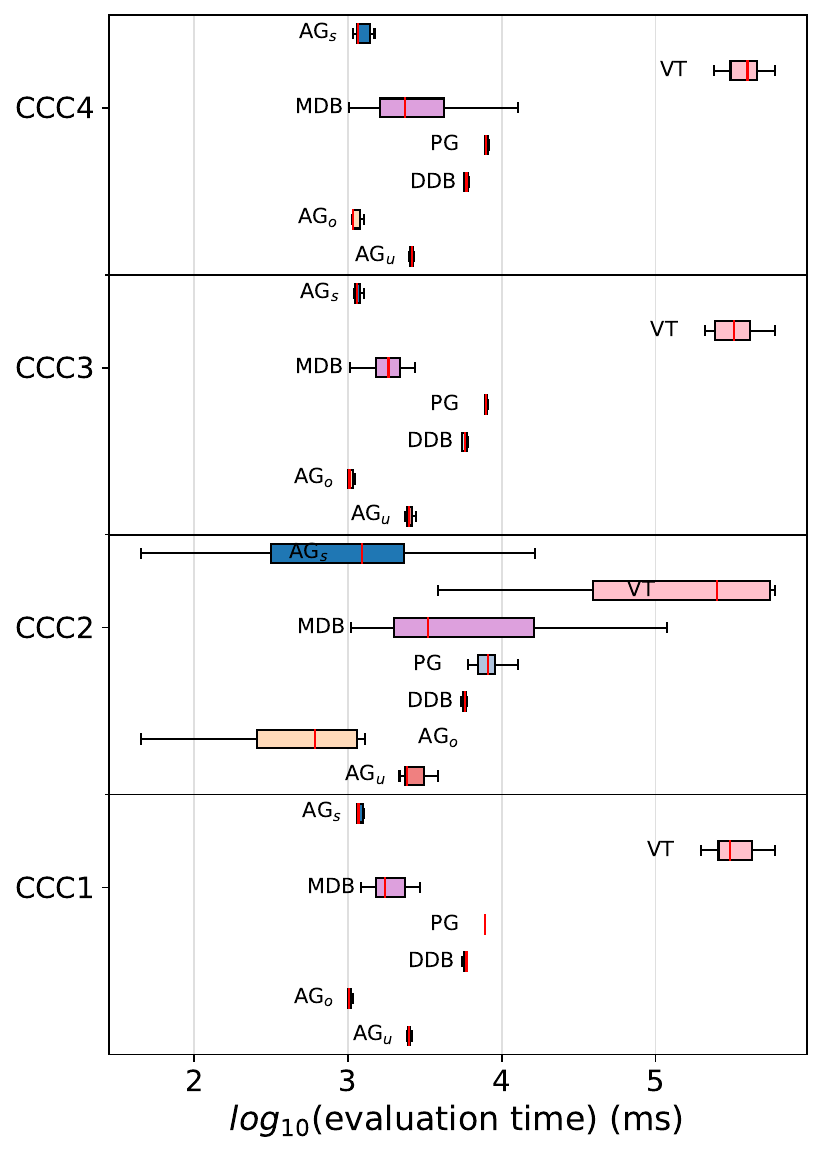}
         \label{fig:context_dbpedia_ccc_rq}
     \end{subfigure}
     \begin{subfigure}[t]{0.32\textwidth}
         \includegraphics[width=\linewidth]{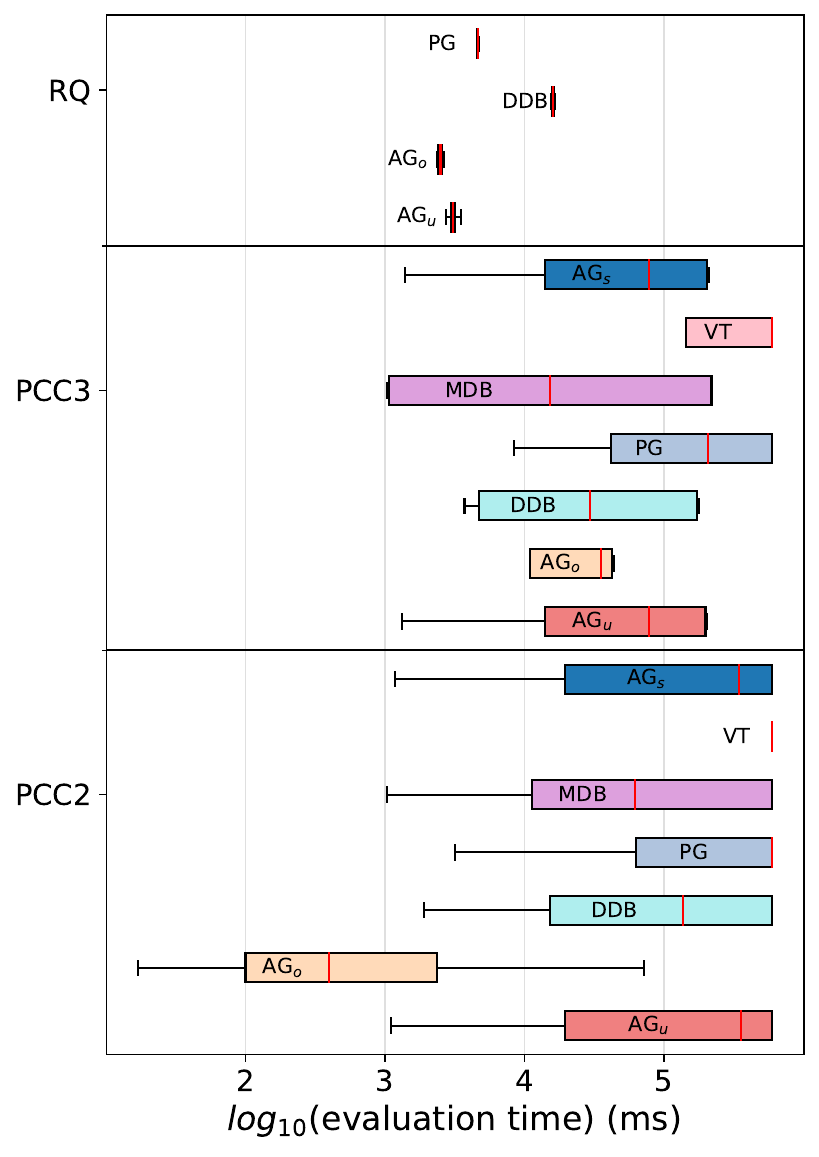}
         \label{fig:context_dbpedia_pcc}
     \end{subfigure}
     \caption{Query evaluation time on \stringdb{} (left) and \dbpedia{} (middle \& right).
     Significant speed-ups of up to several orders of magnitude are realised for most templates (\autoref{sec:experiments-systems}).}
     \label{fig:context}
\end{figure*}

\subsubsection{The proposed optimizations can be applied effectively}
\label{sec:experiments-effective}
The effective application of the proposed optimizations depends on two things. First, application of the optimizations should not be a detriment to the performance of the query optimization procedure itself. Second, the improvement to query evaluation performance achieved through cost-based optimization should be significant. 
\paragraph{The impact on optimization time is modest.}
We report that $\timeplan{\ebpo}$ is below $50$ milliseconds for all query instances of all query templates on both datasets. In other words, all queries considered in this work can be optimized very quickly using the described procedure. To gain further insights into scalability we also plot $\timeplan{\ebpo}$ for particular query \textit{shapes} at various sizes in \autoref{fig:scalability}.

\begin{figure}[ht]
    \centering
     \begin{subfigure}[t]{0.64\columnwidth}
         \centering
         \includegraphics[width=\linewidth]{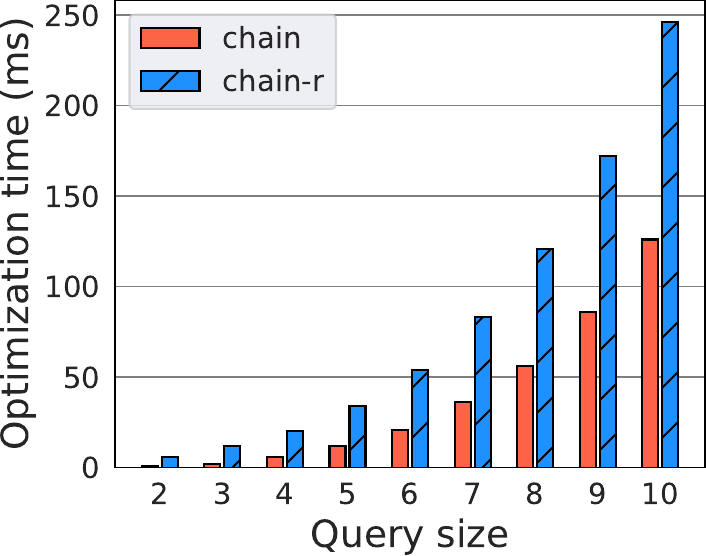}
         \label{fig:scaling_quad}
     \end{subfigure}
     \hfill
     \begin{subfigure}[t]{0.64\columnwidth}
         \centering
         \includegraphics[width=\linewidth]{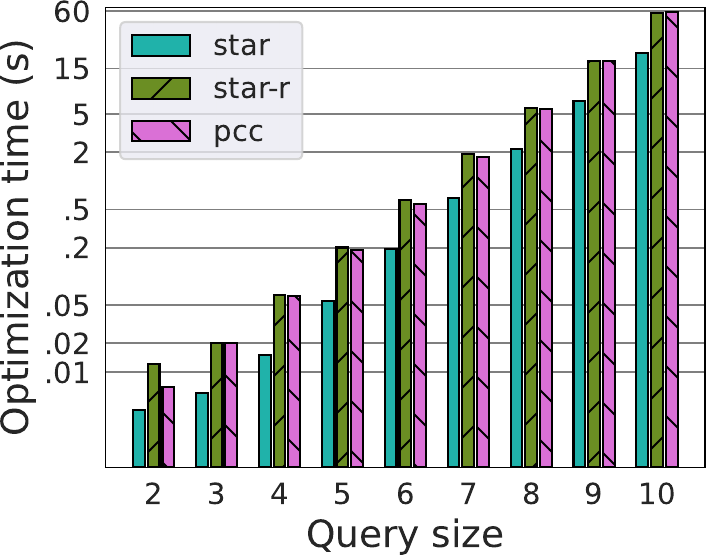}
         \label{fig:scaling_exp}
     \end{subfigure}
     \caption{The optimization time scales well in terms of query shape and size (\autoref{sec:experiments-effective})}
     \label{fig:scalability}
\end{figure}

We consider two well-known query shapes that are extremely common in practice, namely a \textit{chain} query and a \textit{star} query \cite{DBLP:journals/vldb/BonifatiMT20}. For each shape, we consider a version consisting of $n$ non-recursive terms and one (suffixed with \quotes{-r}) consisting of $n$ recursive terms, for values of $n$ from $2$ to $10$. For example, by \dquotes{chain 3} we mean a query of three non-recursive terms arranged in a chain pattern, and by \dquotes{star-r 5} we mean a query of five recursive terms arranged in a star pattern. Additionally, we consider versions of the CCC template but with up to ten terms. We plot the time spent on optimization for each combination of query shape and number of terms. All times are averaged over five runs. We observe that the optimization times for the chain and chain-r shapes increase very slowly as $n$ increases. This is as expected because there exist relatively few viable join orders that avoid Cartesian products for chain queries. The optimization times for the star, star-r and CCC shapes increase at a significantly higher rate, roughly tripling with each increment of $n$ for the larger values of $n$. This too is expected, because the join graph for each of these shapes is a clique and so the number of join orders increases exponentially. We conclude that the enumerator's design scales well since relatively \quotes{long} chain-queries can be optimized in far less than a second, and that even star-queries of up to six recursive terms can be optimized in less than one second.

\paragraph{The improvements to query evaluation performance are significant.} The actual improvements (i.e., those achieved by the estimated best optimized plan $\ebpo$ relative to $\abpu$ the best unoptimized plan in practice) in terms of cardinality (\actcard{}) and evaluation time (\acttime{}) are tabulated in \autoref{tab:improvements_stringdb} and \autoref{tab:improvements_dbpedia}. We emphasize that the actual improvements obtained are \textit{minimal} and \textit{conservative}, due to taking the performance of the best unoptimized plan in practice as a baseline.

We first consider the results on \stringdb{}. The median improvements in terms of query evaluation time range from modest (1.11x) to significant (3.09x) and at its most extreme, a 59-fold speed-up is achieved for the CCC template. We must note that $\abpu$ outperforms $\ebpo$ on a few instances (i.e., either \actcard{} or \acttime{} is below one). In such cases, the optimization procedure simply ends up picking a worse plan than $\abpu$. However, this happens only in a handful of instances.

In summary, the results on \stringdb{} show significant speed-ups in terms of actual query evaluation time up to an order of magnitude while also showing that even larger improvements are possible due to the discrepancies between potential- and actual improvements. In other words, by improving the optimizer's ability to apply the proposed optimizations more effectively (i.e., to make better cost-based decisions) more gains can be made.

On the \dbpedia{} dataset even larger improvements are achieved. The median improvements in terms of query evaluation time range from significant (1.2x) to extreme (42.9x) and at its maximum, a 1480-fold improvement is achieved for the PCC2 template. Again, we note that $\ebpo$ is sometimes outperformed by $\abpo$ but this effect is milder compared to \stringdb{}. Regarding the \quotes{gap} between potential- and actual improvement, we note a particularly large gap for the PCC3 template in \dbpedia{}. We believe that this is, at least in part, due to the greater complexity of these templates. The potential improvement for PCC3 on \stringdb{} is much lower than it is on \dbpedia{} making this less of an issue.

In summary, we observe significant improvements across all templates in query evaluation performance on DBPedia based on the \actcard{} and \acttime{} metrics, of up to three orders of magnitude. We emphasize again that the actual improvements obtained are \textit{minimal} and therefore conservative, due to taking the performance of the best unoptimized plan in practice as a baseline.
\begin{table}[ht]
\centering
\resizebox{0.96\columnwidth}{!}{
\begin{tabular}{ll||ll||llll}
\multicolumn{2}{l||}{Dataset} & \multicolumn{2}{l||}{DBPedia} & \multicolumn{4}{l}{STRING} \\
\multicolumn{2}{l||}{Template} & \multicolumn{2}{l||}{PCC2 (\#57)} & \multicolumn{2}{l}{PCC2 (\#2)} & \multicolumn{2}{l}{PCC3 (\#13)} \\
\multicolumn{2}{l||}{Metric} & $t(\abpo)$ & $t(\ebpo)$ & $t(\abpo)$ & $t(\ebpo)$ & $t(\abpo)$ & $t(\ebpo)$ \\ \hline \hline
\multicolumn{2}{l||}{Min} & 2 & \cellcolor[HTML]{FFF2E3}84 & 69489 & \cellcolor[HTML]{DAE8FC}78798 & 1983 & \cellcolor[HTML]{DAE8FC}2051 \\ \hline
 & 10 & 5 & \cellcolor[HTML]{FFF2E3}86 & 69489 & \cellcolor[HTML]{DAE8FC}78798 & 11926 & \cellcolor[HTML]{DAE8FC}12890 \\
 & 25 & 12 & \cellcolor[HTML]{FFF2E3}99 & 69489 & \cellcolor[HTML]{DAE8FC}78798 & 17637 & \cellcolor[HTML]{DAE8FC}23650 \\
 & 50 & 166 & \cellcolor[HTML]{FFF2E3}\uline{196 } & 69489 & \cellcolor[HTML]{DAE8FC}\dashuline{78798 } & 31673 & \cellcolor[HTML]{DAE8FC}\dashuline{72095 } \\
 & 75 & 1399 & \cellcolor[HTML]{FFF2E3}1534 & 1.2E5 & \cellcolor[HTML]{DAE8FC}\dotuline{1.2E5 } & 70377 & \cellcolor[HTML]{DAE8FC}\dotuline{1.2E5 } \\
\multirow{-5}{*}{\rotatebox{90}{Percentile}} & 90 & 11368 & \cellcolor[HTML]{FFF2E3}20108 & 1.2E5 & \cellcolor[HTML]{DAE8FC}1.2E5 & 94376 & \cellcolor[HTML]{DAE8FC}1.2E5 \\ \hline
\multicolumn{2}{l||}{Max} & 6E5 & \cellcolor[HTML]{FFF2E3}\dotuline{6E5 } & 1.2E5 & \cellcolor[HTML]{DAE8FC}1.2E5 & 110333 & \cellcolor[HTML]{DAE8FC}1.2E5 \\
\multicolumn{2}{l||}{Mean} & 4640 & \cellcolor[HTML]{FFF2E3}6331 & 94745 & \cellcolor[HTML]{DAE8FC}99399 & 46865 & \cellcolor[HTML]{DAE8FC}70589
\end{tabular}
}
\caption{Evaluation times for instances where all \textit{unoptimized} plans timed out.
Optimization yields \uline{sub-second} evaluation for many \dbpedia{} instances and \dashuline{in-time} evaluation for many \stringdb{} instances. Other instances still \dotuline{time out}. The number of instances is denoted (\#$n$) (\autoref{sec:experiments-potency},\autoref{sec:experiments-effective}).}
\label{tab:timeout}
\end{table}

\subsubsection{Query evaluation performance in context}
\label{sec:experiments-systems}
\autoref{fig:context} plots the distributions of query evaluation times using each system's \textit{estimated} best plan across all mined instances, templates and systems in order to place \avantgraph{}'s performance in context. We observe significant speed-ups of up to several orders of magnitude in terms of median evaluation time for most templates across the two datasets. The exception being the PCC3 template. For the DBPedia dataset, the median performance of \avantgraphs{} lags behind that of \duckdb{} on that template. For the \stringdb{} dataset, it lags behind \postgres{}, and on both it lags behind \millenniumdb{}. The primary reason for this appears to be a lack of quality in cost-estimation. That is, based on \autoref{tab:improvements_stringdb} and \autoref{tab:improvements_dbpedia} we know that there exist plans for PCC3 which perform significantly better (up to two order of magnitude better, for DBPedia) than the actual plan being chosen by the optimizer. This can be observed from the difference between the potential improvements (PT and PC) compared to the actual improvements (AT and AC).
Note that the performance of \agwg{} typically lies in-between that of \agu{} and \ags{} as can be expected since \ags{} subsumes \agwg{} which subsumes \agu{} in turn. For some templates, such as CCC1, CCC3 and CCC4 on \dbpedia{} the differnece between \agwg{} and \agu{} is slight (i.e., seeding of interior closures and stacking does not appear to have much of an effect here). Conversely, for the CCC2 and PCC2 templates on \dbpedia{} and all templates on \stringdb{} the difference is far more pronounced (i.e., these novel seeding-based optimizations make a significant difference)

\section{Related work}
\label{sec:related}
Navigational queries are a key class of queries for graph database systems, and while the problem of optimizing their evaluation has received considerable attention in recent years, many systems still struggle with the evaluation of such queries in practice \cite{wdbench}.

The most basic building blocks of navigational graph queries, i.e., path queries in the form of RPQs and sub-graph matching queries have received considerable attention on their own. The research on the optimization of RPQs is characterized by the application of indexing- ~\cite{arroyuelo2023optimizing,DBLP:conf/icde/ArroyueloHNR22,DBLP:conf/icde/MeimarisPMA17,DBLP:conf/grades/TetzelVPL17,DBLP:conf/sigmod/GubichevBS13} and other graph-compression techniques~\cite{DBLP:conf/icde/NaMYWH22}. \added[id=R1,comment={D4}]{Research on sub-graph matching queries has instead focussed on join (order) optimization~\cite{aimonier23,Arch22Souffle} and the application of worst-case optimal join algorithms~\cite{DBLP:conf/semweb/HoganRRS19,Nguyen15WCO}.}

While evaluating navigational queries certainly introduces a set of challenges, it also introduces several opportunities. One such opportunity is the concept of \textit{seeding}, as introduced by Waveguide \cite{DBLP:conf/sigmod/YakovetsGG16}. Waveguide allows for transitive closures to be computed from a set of vertices that is \textit{constrained} by other parts of the query, such as joins or filters, thus potentially reducing the total number of tuples manipulated in the evaluation of the transitive closure. A drawback of Waveguide is its representation of query plans as \textit{automata}, for which it is practically infeasible to integrate them within the more conventional \textit{tree-based} model of query plans in use by the majority of query evaluation engines.

The $\mu$-RA introduced in \cite{DBLP:conf/sigmod/JachietGGL20} generalizes some of Waveguide's optimization techniques, and applies them to UCRPQs rather than RPQs. It extends the Relational Algebra (RA) with a \textit{fix-point} operator $\mu$ that is amendable to having various operations such as \mbox{(anti-)projection}, filters, joins and even other fix-points \textit{pushed} down to it. These transformations are presented algebraically and their application is performed in an ad-hoc manner, that is not guaranteed to consider all semantically equivalent query plans induced by the set of transformations.

We improve upon Waveguide specifically by presenting a unified, \textit{graph-based} model for query plans that can capture all optimization techniques previously represented by automata while generalizing the more conventional tree-based model. We improve upon the $\mu$-RA specifically by addressing the problem of enumerating query plans using a \textit{constructive, rule-based} and \textit{top-down} approach that offers modularity and extensibility. We improve upon both existing approaches by introducing novel forms of seeding and by considering RQs rather than RPQs or UCRPQs.

Our top-down, rule-based approach to query optimization bears some similarity to more general program optimization techniques such as equality saturation in frameworks such as Egg \cite{eggopt}.
Because these techniques and frameworks apply to program optimization in general, rather than query optimization in particular, there is no assumption of optimal sub-structure regarding the composition of solutions based only on the \textit{optimal} solutions to sub-problems. As a consequence, such approaches are infeasible for the large search spaces that are typically encountered in query optimization. A further, shallow similarity exists between seeding-based optimization and \textit{demand transformation} in Datalog programs. This similarity extends only so far as the use of \textit{constants} in Datalog programs is concerned, where it leverages them in order to avoid redundant facts being derived during evaluation~\cite{DBLP:journals/tplp/TekleL16}.

On a meta-level, where the improvements of a query optimizer is concerned (i.e., the optimization of an optimizer), there is existing related work on the join-ordering problem in SPARQL queries~\cite{DBLP:conf/edbt/Gubichev014} based on \textit{characteristic sets} which were originally developed in the context of cardinality estimation~\cite{DBLP:conf/icde/NeumannM11}.

\section{Conclusions}
\label{sec:conclusion}

We have presented a comprehensive study of a novel approach to optimization of navigational graph queries. 
In addition to strictly subsuming known optimization techniques, our method introduces a number of novel powerful optimizations with the goal to effectively enable practical evaluation of complex navigational queries.
We have shown that our approach is able to significantly (up to several orders of magnitude) outperform state-of-the-art query evaluation techniques on a wide range of queries on diverse data sets.
We have shown both analytically and empirically that our approach is able to achieve this performance improvement at a small additional cost of query optimization time.
Future work includes improving cost estimation and extending our approach to support more expressive navigational query languages.
\begin{figure}[ht]
     \centering
     \begin{subfigure}{0.97\columnwidth}
         \centering
         \includegraphics[width=\linewidth]{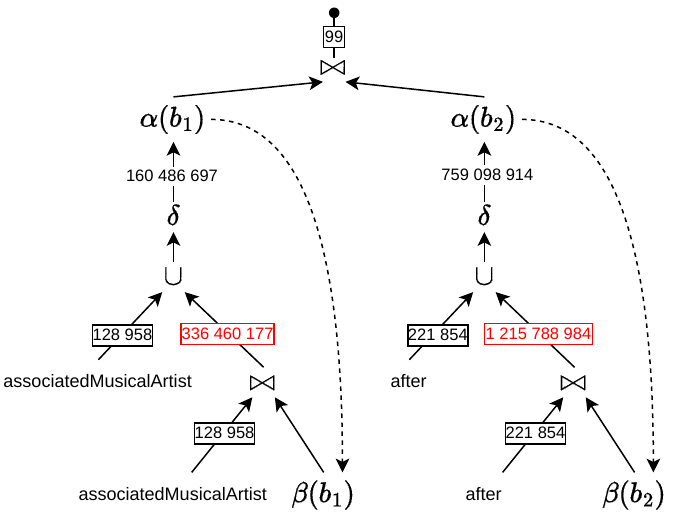}
         \label{fig:pcc2_sota_plan}
     \end{subfigure}
     \hfill
     \begin{subfigure}{0.97\columnwidth}
         \centering
         \includegraphics[width=\linewidth]{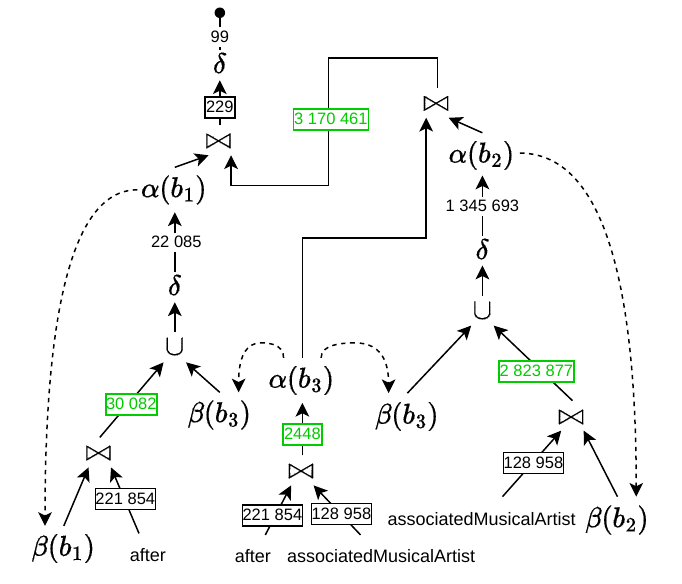}
         \label{fig:pcc2_novel_plan}
     \end{subfigure}
    \caption{Plans $\abpu$ (top) and $\abpo$ (bottom) for a PCC2 instance on \dbpedia{} (\autoref{sec:use_case_study})}
    \label{fig:pcc2_plans}
\end{figure}

\appendix

\section{DBPedia use-case Study}
\label{sec:use_case_study}

To study the potency of seeding interior closures in more detail, we consider an instance $Q$ of query template PCC2 on DBPedia with $l_1$ and $l_2$ set to \dquotes{associatedMusicalArtist} (AMA) and \dquotes{after} (AFT).
This instance asks for pairs of people for whom there exist two different paths (i.e., paths composed of edges with different labels). Edges labeled AMA capture loose associations between musicians, whereas edges labeled AFT capture a generic notion of succession between people, to some position or role. Hence, this instance asks for those pairs of people that have some (in)direct musical association and the latter (in)directly succeeded the former in some capacity. One embedding of this query instance in DBPedia, including the intermediary \quotes{hidden} vertices in each path, is displayed in \autoref{fig:pcc2_embedding}. This embedding captures the fact that Chad Channing, Dale Crover, Dan Peters and Dave Grohl are successive drummers for the band Nirvana and that there exists a sequence of loose musical associations between Chad and Dave spanning fifteen other musicians.

We exhaustively enumerate the plan spaces $\uq$ and $\sq$ and execute all plans within them. This reveals the actual best plans $\abpu \in \uq$ and $\abpo \in \sq$. These plans are displayed in \autoref{fig:pcc2_plans}. Most edges between operators are annotated with the cardinality of the corresponding intermediate result. Annotations with solid borders are used for the edges outgoing from read- and join operators. These cardinalities contribute to the total number of tuples  processed. The other cardinalities (i.e., without a border) do not contribute.

The \textit{seeding query}'s cardinality in $\abpo$ is only 2448, whereas the cardinalities of the sets of AMA and AFT edges are $128\,958$ and $221\,854$. Because both closures in $\abpo$ are computed from a seed, the result sizes for the seeded closures drop from hundreds of millions to millions for the AMA closure and from billions to tens of thousands for the AFT closure. 

The total number of tuples processed is $1\,552\,950\,884$ for $\abpu$ while it is merely $6\,728\,721$ for $\abpo$. This results in an evaluation time of $441\,674$ milliseconds (almost 7.5 minutes) for $\abpu$ and $1225$ milliseconds for $\abpo$.
\begin{figure}[ht]
    \centering
    \includegraphics[width=0.88\columnwidth]{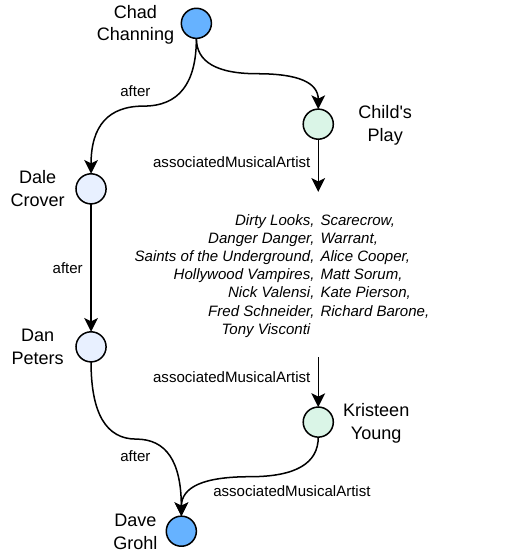}
    \caption{An embedding of PCC2 in DBPedia.}
    \label{fig:pcc2_embedding}
\end{figure}

\begin{acknowledgements}
This project has received funding from the European Union's Horizon Europe framework programme under grant agreement No.~101058573.
\end{acknowledgements}

\bibliographystyle{spmpsci}
\bibliography{bibliography}

\end{document}